\documentclass[pra,twocolumn,aps,floatfix,showpacs,tightenlines,superscriptaddress,amsmath,amssymb,showkeys,10pt,nofootinbib]{revtex4-2}
\pdfoutput=1
\usepackage{comment}
\usepackage{cancel}
\usepackage{graphicx} % Required for inserting images
\usepackage{amsmath,mathtools}
\usepackage{xcolor}
\usepackage{amssymb}
\usepackage{amsfonts, amsthm}
\usepackage{dsfont}
\usepackage{simplewick}
\usepackage[colorlinks=true, citecolor=blue, urlcolor=blue]{hyperref}  

\newtheorem{proposition}{\bf Proposition}[section]
\newtheorem{theorem}{\bf Theorem}[section]

\newtheorem{corollary}{\bf Corollary}[section]
\newtheorem{definition}{\bf Definition}[section]
\newtheorem{lemma}{\bf Lemma}[section]

\makeatletter
\def\smalloverbrace#1{\mathop{\vbox{\m@th\ialign{##\crcr\noalign{\kern3\p@}%
  \tiny\downbracefill\crcr\noalign{\kern3\p@\nointerlineskip}
  $\hfil\displaystyle{#1}\hfil$\crcr}}}\limits}
\makeatother

\def\choi{\mathsf{Ch}}

\def\Tr{\operatorname{Tr}}

\def\dim{\operatorname{dim}}
\def\supp{\operatorname{supp}}
 
\def\>{\rangle} \def\<{\langle} 
 
\def\Ch{\mathsf{Ch}}

\def\T+{\mathsf{T}_+} \def\ETy{\set{EleTypes}}
\def\Ty{\set{Types}} 
\def\Ev#1{\set{T}_{\mathbb R}(#1)}
\def\Evd#1{\set{T}_1(#1)}

\def\para{\parallel}

\mathchardef\mhyphen="2D

\newcommand{\ket}[1]{|#1\rangle}

\newcommand{\Bra}[1]{\langle \! \langle#1|}
\newcommand{\Ket}[1]{|#1\rangle \! \rangle}

\newcommand{\KetBra}[2]{{\Ket{#1}\Bra{#2}}}
\newcommand{\hilb}[1]{\mathcal{#1}}
\newcommand{\set}[1]{{\sf #1}}

\newcommand{\dket}[1]{| #1 \rangle\!\rangle} \newcommand{\dbra}[1]{\langle\!\langle #1 |}

\newcommand{\bigast}{\mathop{\scalebox{1.5}{\raisebox{-0.2ex}{$\ast$}}}}%

\begin{document}

\title{Higher-order transformations of bidirectional quantum processes}

\author{Luca Apadula}
\email{luca.apadula@cea.fr}
\affiliation{Universit\'e Paris-Saclay, 3 rue Joliot Curie, Breguet Building, 91190 Gif-sur-Yvette, France}
\affiliation{Commissariat \`a l'\'energie atomique et aux \'energies alternatives (CEA), DSM/LARISM, Orme des Merisiers, 91190 Gif-sur-Yvette, France}

\author{Alessandro Bisio}
\email{alessandro.bisio@unipv.it}
\affiliation{Dipartimento di Fisica, Universit\`a di Pavia, via Bassi 6, 27100 Pavia, Italy}
\affiliation{Istituto Nazionale di Fisica  Nucleare, Sezione di Pavia, Italy}

\author{Giulio Chiribella}
\email{giulio@cs.hku.hk}
\affiliation{QICI Quantum Information and Computation Initiative, Department of Computer Science, The University of Hong Kong, Pokfulam Road, Hong Kong}
\affiliation{Quantum Group, Department of Computer Science, University of Oxford, Wolfson Building, Parks Road, Oxford, OX1 3QD, United Kingdom}
\affiliation{Perimeter Institute for Theoretical Physics, 31 Caroline Street North, Waterloo, Ontario, Canada}

\author{Paolo Perinotti}
\email{paolo.perinotti@unipv.it}
\affiliation{Dipartimento di Fisica, Universit\`a di Pavia, via Bassi 6, 27100 Pavia, Italy}
\affiliation{Istituto Nazionale di Fisica Nucleare, Sezione di Pavia, Italy}

\author{Kyrylo Simonov}
\email{kyrylo.simonov@univie.ac.at}
\affiliation{Fakultät für Mathematik, Universität Wien, Oskar-Morgenstern-Platz 1, 1090 Vienna, Austria}

\begin{abstract}
Bidirectional devices are devices for which the roles of the input and output ports can be exchanged. Mathematically, these devices are described by bistochastic quantum channels, namely completely positive linear maps that are both trace-preserving and identity-preserving. Recently, it has been shown that bidirectional quantum devices can, in principle, be used in ways that are incompatible with a definite input-output direction, giving rise to a new phenomenon called input-output indefiniteness. Here we characterize the most general forms of input-output indefiniteness, associated with a hierarchy of higher-order transformations built from transformations of bistochastic quantum channels. Some levels of the hierarchy correspond to transformations that combine bistochastic channels in a definite causal order, while generally using each channel in an indefinite input-output direction. For other levels of the hierarchy, the indefiniteness can involve both the local input-output direction of each process and the global causal order among the processes. On the foundational side, the hierarchy of higher-order transformations characterized here can be regarded as the largest set of physical processes compatible with a time-symmetric variant of quantum theory, where the possible state transformations are restricted to bistochastic channels.
\end{abstract}

\maketitle

% --------------------------------------------------------------------------
% Sec. I
% --------------------------------------------------------------------------

\section{Introduction}

An important class of quantum devices is bidirectional, meaning that the roles of the input port (through which the system enters) and the output port (from which the system exits) are interchangeable. For example, optical crystals such as half-wave plates and quarter-wave plates can be traversed in two opposite directions, generally giving rise to different rotations of the photon's polarization. Mathematically, bidirectional devices correspond to bistochastic quantum channels \cite{landau1993birkhoff}, that is, completely positive linear maps that are both trace-preserving and identity-preserving. These channels play a central role in quantum thermodynamics \cite{gour2015resource, chiribella2017microcanonical} and in the theory of majorization \cite{gour2018quantum}. Recently, they have attracted increasing interest in the foundations of quantum mechanics, due to their connection with the problem of time symmetry \cite{di2021arrow, hardy2021time, chiribella2021symmetries, chiribella2022quantum}.  

By definition, a bidirectional quantum device can be used in two alternative ways, depending on which port is treated as the input port. Recently, it was shown that bidirectional quantum devices can also be used in new ways, in which it is impossible to determine which port is the input port and which port is the output port \cite{chiribella2022quantum}. For example, one can construct interferometric setups in which a quantum system travels through the device in a coherent superposition of two opposite directions, giving rise to a new primitive called the ``quantum time flip'' \cite{chiribella2022quantum}. These setups give rise to a new quantum feature known as input-output indefiniteness \cite{Guo2024}, with applications to quantum channel discrimination \cite{chiribella2022quantum}, quantum communication \cite{Liu_2023}, quantum metrology \cite{Guo2024, xia2024nanoradian, agrawal2025indefinite}, quantum thermodynamics \cite{rubino2021quantum}, quantum memory effects \cite{Karpat2024}, and causal inequalities \cite{Liu2025}. Input-output indefiniteness has also been experimentally demonstrated in recent photonic experiments \cite{Guo2024, stromberg2024experimental}. 

The basic framework for quantum operations with indefinite input-output direction was developed in Ref.~\cite{chiribella2022quantum}. Mathematically, these operations correspond to transformations of bistochastic channels; their structure was characterized in \cite{chiribella2022quantum} in the case of transformations mapping $N$ bistochastic channels into a single (not necessarily bistochastic) channel. However, the case of general higher-order transformations with arbitrary numbers of channels in the output has remained open until now. These transformations represent deeper layers in the manipulation of bistochastic devices and can potentially provide new insights into time symmetry in quantum mechanics. 

In this paper, we provide a complete characterization of all admissible transformations with indefinite input-output direction. We define an infinite hierarchy of higher-order transformations, in which transformations of bistochastic channels provide the first level, and transformations at lower levels serve as inputs for transformations at higher levels. By characterizing the various levels of the hierarchy, we identify a variety of processes with interesting causal structures. For example, one of the levels of the hierarchy contains transformations that combine bistochastic channels in a definite causal order, while using each individual channel in an indefinite input-output direction. Physically, these transformations can be thought of as ways of connecting the operations performed in a sequence of local laboratories, in which the input-output direction is not fixed, but the causal order of the laboratories is. For these transformations, we prove a realization theorem, showing that they can be decomposed into a causally ordered sequence of transformations of individual bistochastic channels.

Technically, our results build on the notion of quantum supermap~\cite{chiribella2008transforming, chiribella2009theoretical, chiribella2013quantum} and on the hierarchy theory of higher order quantum transformations developed in~\cite{doi:10.1098/rspa.2018.0706, Apadula2024} (see also~\cite{bisioacta, taranto2025higherorderquantumoperations} for a broad overview). These concepts and methods have broad applicability to quantum information,  where they provide a flexible toolbox for the optimization of quantum protocols \cite{chiribella2008quantum, bisio2014optimal, chiribella2016optimal}. Notable applications include the study of  quantum channels with memory \cite{kretschmann2005quantum}, interactive quantum protocols~\cite{Gutoski:2007:TGT:1250790.1250873}, quantum metrology~\cite{chiribella2008memory, chiribella2012optimal, liu2023optimal}, quantum channel discrimination~\cite{chiribella2008memory, chiribella2012perfect, bavaresco2022unitary}, channel tomography~\cite{PhysRevLett.102.010404, li2025quantum}, cloning of transformations and measurements~\cite{chiribella2008optimal, bisio2011cloning, dur2015deterministic, PhysRevLett.114.120504}, storage and retrieval of transformations~\cite{bisio2010optimal, mozrzymas2018optimal, sedlak2019optimal, PhysRevA.102.032618}, as well as inversion, conjugation, and transposition of unitary operations~\cite{bisio2014optimal, Quintino2022, PhysRevLett.131.120602, chen2024quantum}. Moreover, higher order maps are central to the study of causality in quantum mechanics, where they provide a framework for quantum causal networks~\cite{chiribella2009theoretical,costa2016quantum, allen2017quantum}, quantum causal discovery~\cite{ried2015quantum, chiribella2019quantum, bai2022quantum}, as well as the stepping stone for the study of indefinite causal order~\cite{chiribella2009beyond, oreshkov2012quantum, chiribella2013quantum, AraujoLugano2014, Baumeler_2016, Wechs2021, Milz_2024, steffinlongo2025}.  
 
Quantum supermaps were originally defined as transformations of quantum channels, that is, completely positive trace-preserving maps. The corresponding hierarchy of higher-order transformations was then built by recursively defining higher-level transformations starting from quantum channels at the basic level \cite{Perinotti2017, doi:10.1098/rspa.2018.0706}. This hierarchy represents the largest set of processes that are logically consistent with the standard operational formulation of quantum theory.

The key difference here is that we consider supermaps defined on bistochastic channels, a proper subset of the set of all quantum channels. From a foundational point of view, bistochastic channels can be viewed as the dynamics in a time-symmetric variant of quantum theory, where the allowed state preparations are ensemble decompositions of the maximally mixed state \cite{chiribella2021symmetries}. The hierarchy characterized in this paper then represents the largest set of logically conceivable processes that are, in principle, compatible with this time-symmetric variant of quantum theory.

The remainder of the paper is organized as follows. In Section~\ref{sec:Choi}, we review the Choi isomorphism and define bistochastic transformations. In Section~\ref{sec:types}, we review the type system of quantum higher-order transformations. In Section~\ref{sec:axioms}, we present the operational axioms of higher-order quantum theory with indefinite input-output direction and provide a full characterization of deterministic admissible higher-order maps. In Section~\ref{sec:processes}, we present several applications of the developed theory. In particular, we introduce hierarchies of networks of higher-order maps and generalized process matrices, with bistochastic generalizations of quantum combs and process matrices arising as particular cases. In Section~\ref{sec:conclusions}, we draw conclusions and outlook.

% --------------------------------------------------------------------------
% Sec. II
% --------------------------------------------------------------------------

\section{Quantum theory and the Choi isomorphism}
\label{sec:Choi}

We denote quantum systems by capital letters $A, B, \dots, Z$, each associated with a Hilbert space $\hilb{H}_A, \hilb{H}_B, \dots, \hilb{H}_Z$. All systems considered are finite-dimensional, i.e., $d_A \coloneqq \dim(\hilb{H}_A) < \infty$. The one-dimensional (trivial) system is denoted by $I$, meaning $\hilb{H}_I = \mathbb{C}$. The identity operator on $\hilb{H}_A$ is denoted by $\mathds{1}_A$.

Given two systems $A$ and $B$, their parallel composition is denoted by $AB$, with joint space $\hilb{H}_{AB} =
\hilb{H}_{A} \otimes \hilb{H}_{B}$. We denote:
\begin{align*}
\mathcal{L}(\hilb{H}_A) &\equiv \text{linear operators on } \hilb{H}_A, \\
\mathsf{Hrm}(\hilb{H}_A) &\equiv \text{Hermitian operators on } \hilb{H}_A, \\
\mathsf{Trl}(\hilb{H}_A) &\equiv \text{Hermitian traceless operators on } \hilb{H}_A,
\end{align*}
and $\mathcal{L}(\mathcal{L}(\hilb{H}_A), \mathcal{L}(\hilb{H}_B))$ the set of linear maps from $\mathcal{L}(\hilb{H}_A)$ to $\mathcal{L}(\hilb{H}_B)$.

A \emph{state} of a quantum system $A$ is a positive operator $ 0 \leq \rho \in \mathcal{L}(\hilb{H}_A)$ such that $\Tr[\rho] \leq 1$. We notice that our definition includes \emph{probabilistic} states, i.e., elements of a quantum ensemble. If $\Tr[\rho] = 1$ for a quantum state $\rho$, we say that $\rho$ is a deterministic state. We denote by $T(A)$ the set of probabilistic states of system $A$ and by $\mathsf{T}_1(A)$ the set of deterministic states of system $A$.

Following Kraus' axiomatic approach \cite{Kraus}, a transformation $\mathcal{M} \in \mathcal{L}(\mathcal{L}(\hilb{H}_A), \mathcal{L}(\hilb{H}_B))$ from system $A$ to system $B$ is \emph{admissible} if it is a completely positive and trace-non-increasing map from quantum states to quantum states, and we say that $\mathcal{M}$ is a \emph{quantum operation}. A quantum operation that is trace preserving is called a \emph{quantum channel}. We denote by $T(A \to B)$ the set of quantum operations from system $A$ to system $B$ and by $\mathsf{T}_1(A \to B)$ the set of quantum channels from $A$ to $B$.

A set $\{ \mathcal{M}_i\}_{i\in \set S}$ (where $\set S$ is a set of indices) of quantum operations from system $A$ to system $B$ such that $\mathcal{M} \coloneqq \sum_{i \in \set S} \mathcal{M}_i$ is trace preserving is called a \emph{quantum instrument}. A special instance of an instrument is given by positive-operator-valued measures (POVMs), which map states into probabilities and are described by a collection of positive operators that sum to the identity. Moreover, states of a quantum system $A$ can be considered as a special case of completely positive maps from the trivial system $I$ to $A$.

A convenient way to represent maps between operator spaces is given by the Choi isomorphism \cite{CHOI1975285}.

\begin{theorem}[Choi isomorphism]\label{thm:choi}
Consider the map $\choi:\mathcal{L}(\mathcal{L}(\hilb{H}_A),\mathcal{L}(\hilb{H}_B)) \to \mathcal{L}(\hilb{H}_B \otimes \hilb{H}_A)$ defined as 
\begin{align} 
    \choi:\mathcal{M} \mapsto M, \qquad M\coloneqq(\mathcal{I}_{A}\otimes \mathcal{M}) [\KetBra{\mathds{1}}{\mathds{1}}],
\end{align} 
where $\mathcal{I}_{A}$ is the identity map on $\mathcal{L}(\hilb{H}_A)$ and $\Ket{\mathds{1}}\coloneqq\sum_{n=1}^{d_A} \ket{n}\ket{n}$, $\{ \ket{n} \}_{n=1}^{d_A}$ denoting an orthonormal basis of $\hilb{H}_A$. Then $\choi$ defines an isomorphism between $\mathcal{L}(\mathcal{L}(\hilb{H}_A),\mathcal{L}(\hilb{H}_B))$ and $\mathcal{L}(\hilb{H}_A \otimes \hilb{H}_B)$. The operator $M = \choi(\mathcal{M})$ is called the \emph{Choi operator} of $\mathcal{M}$. Moreover, one has\footnote{$\Tr_B$
denotes the partial trace on system $B$ and $\mathds{1}_A$
is the identity operator on system $A$.}:
\begin{align*} 
    &\Tr[\mathcal{M}(X)] = \Tr[X] \quad \forall X \in \mathcal{L}(\hilb{H}_A) \, \Leftrightarrow \, \Tr_B[M] = \mathds{1}_A , \\
    &\mathcal{M}[X]^\dagger = \mathcal{M}[X^\dagger] \, \Leftrightarrow \, M^\dagger = M, \\ 
    & \mathcal{M} \mbox{ is completely positive} \; \Leftrightarrow \; M \geq 0.
\end{align*}
\end{theorem}

The inverse of the map $\choi$ is given by the following expression:
\begin{align}\label{eq:invchoi} 
    & (\choi^{-1}(M))[O] = \Tr_A[(O^T\otimes \mathds{1}_B)M], \\
    &O \in \mathcal{L}(\mathcal{H}_A), \quad M\in \mathcal{L}(\mathcal{H}_A\otimes \mathcal{H}_B), \nonumber
\end{align}
where $O^T$ denotes the transpose operator with respect to the orthonormal basis we used to define $\Ket{\mathds{1}}$ in Theorem \ref{thm:choi}.

A linear operator $M \in \mathcal{L}(\hilb{H}_B \otimes \hilb{H}_A)$ is the Choi operator of a quantum channel $\mathcal{M} : \mathcal{L}(\hilb{H}_A) \to \mathcal{L}(\hilb{H}_B)$, i.e., $M = \choi (\mathcal{M})$, if and only if
\begin{equation}\label{eq:chanChoi}
    M \geq 0, \qquad \Tr_B[M]=\mathds{1}_A.
\end{equation}
Similarly, $O \in \mathcal{L}(\hilb{H}_B \otimes \hilb{H}_A)$ is the Choi operator of a quantum operation $\mathcal{O} : \mathcal{L}(\hilb{H}_A) \to \mathcal{L}(\hilb{H}_B)$ if and only if $O \geq 0$ and there exists a quantum channel $\mathcal{D} : \mathcal{L}(\hilb{H}_A) \to \mathcal{L}(\hilb{H}_B)$ such that $O \leq D$, where $D \coloneqq \choi (\mathcal{D})$.

Given two maps $\mathcal{M} \in \mathcal{L}(\mathcal{L}(\hilb{H}_A), \mathcal{L}(\hilb{H}_B))$ and $\mathcal{N} \in \mathcal{L}(\mathcal{L}(\hilb{H}_B), \mathcal{L}(\hilb{H}_C))$, we can consider their composition $\mathcal{N} \circ \mathcal{M} \in \mathcal{L}(\mathcal{L}(\hilb{H}_A), \mathcal{L}(\hilb{H}_C))$, corresponding to their sequential application. On Choi operators, this physical composition corresponds to the link product, an associative operation~\cite{chiribella2008quantum, chiribella2009theoretical, bisioacta}.
\begin{definition}\label{def:linkProd}
    Let $M \in \mathcal{L}(\mathcal{H}_{\mathsf{I}})$ and $N \in \mathcal{L}(\mathcal{H}_{\mathsf{J}})$, where $\mathsf{I}$ and $\mathsf{J}$ are two finite sets of indices, and $\mathcal{H}_{\mathsf{A}}\coloneqq\otimes_{i \in \mathsf{A}} \mathcal{H}_i$. Then the link product $N * M$ is an operator in $\mathcal{L}(\mathcal{H}_{\mathsf{I}\setminus \mathsf{J}} \otimes \mathcal{H}_{\mathsf{J}\setminus \mathsf{I}})$ defined as
    \begin{equation}
        N * M = \Tr_{\mathsf{I}\cap \mathsf{J}}[(\mathds{1}_{\mathsf{J}\setminus\mathsf{I}}\otimes M^{T_{\mathsf{I}\cap \mathsf{J}}})(\mathds{1}_{\mathsf{I}\setminus\mathsf{J}}\otimes N)],
    \end{equation}
    where $\mathsf{A} \setminus \mathsf{B} \coloneqq \{i \in \mathsf{A}\mid i \notin \mathsf{B}\}$. In particular, if $\mathsf{I} \cap \mathsf{J} = \emptyset$, then $N * M = N \otimes M$, while if $\mathsf{I} = \mathsf{J}$, then $N * M = \Tr[N^T M]$.
\end{definition}

We will use a calligraphic font to denote a map between operator spaces and a roman font for the corresponding Choi operator. For example, given the linear map $\mathcal{M} : \mathcal{L}(\hilb{H}_A) \to \mathcal{L}(\hilb{H}_B)$, $M$ will denote its Choi operator. We will make an implicit use of the Choi isomorphism several times throughout the paper. For example, we will write ``(the linear operator) $R$ is a quantum channel'' instead of ``$R$ is the Choi operator of a quantum channel''. Similarly, the expressions $T(A\to B)$ and $\mathsf{T}_1(A \to B)$ will be used to denote the set of Choi operators of quantum operations and the set of Choi operators of quantum channels, respectively.

% --------------------------------------------------------------------------
\subsection{Bistochastic quantum transformations}

\begin{definition} (Partially bistochastic channel)
    Let ${R} \in \mathsf{T}_1(AB \to CD)$ be a quantum channel from the bipartite system $AB$ to the bipartite system $CD$. We say that ${R}$ is \emph{$A \mhyphen C$ bistochastic} if
    \begin{align}\label{eq:1}
        R \geq 0,\quad \Tr_{CD}[R] = \mathds{1}_{AB} \ \mbox{and}\  \Tr_{AD}[R] = \mathds{1}_{CB}.
    \end{align}
    We denote by ${T}_1(\hat{A}B \to \hat{C}D)$ the set of quantum channels from $AB$ to $CD$ that are $A \mhyphen C$ bistochastic.

    Let ${S} \in {T}(AB \to CD)$ be a quantum operation from the bipartite system $AB$ to the bipartite system $CD$. We say that ${S}$ is $A \mhyphen C$ bistochastic if there exists $R\in {T}_1(\hat{A}B \to \hat{C}D)$ such that $0 \leq S \leq R$. We denote by ${T}(\hat{A}B \to \hat{C}D)$ the set of quantum operations from $AB$ to $CD$ that are $A \mhyphen C$ bistochastic.
\end{definition}

In an \emph{$A \mhyphen C$} bistochastic channel, we can use $A$ as an input and $C$ as an output, or $C$ as an input and $A$ as an output; that is, exchanging $A$ and $C$ leads to a well-defined channel. Therefore,
\begin{align}
  \label{eq:bistochT}
  \begin{aligned}
    T_1&(\hat{A}B \to \hat{C}D) = \\
    &\mathsf{T}_1(AB \to CD) \cap \mathsf{T}_1(CB \to AD).
  \end{aligned}
\end{align}
By condition~\eqref{eq:1}, one can easily prove that the systems $A$ and $C$ must be isomorphic, i.e., $d_A = d_C$. An $A\mhyphen C$ bistochastic channel $R$ will be graphically represented as
\begin{equation*}
 R=\begin{array}{l}
      \includegraphics[width=60pt]{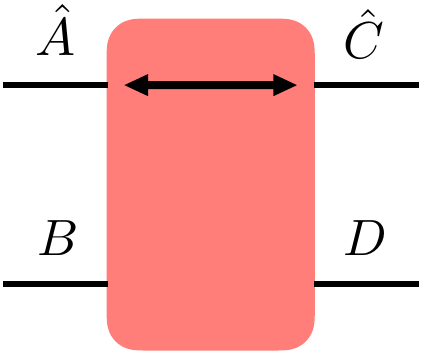},
 \end{array}   
\end{equation*}
where the bidirectional arrow reflects the bistochasticity 
of the corresponding quantum systems. The definition of $A \mhyphen C$ bistochastic channels (and $A \mhyphen C$ bistochastic quantum operations) can be generalized to the case in which the systems $A$ and $C$ are multipartite. For example, imposing bistochasticity on bipartite systems yields
\begin{align}
  \label{eq:6}
  \begin{aligned}
    T_1&(\widehat{AB}C \to \widehat{DE}F) \coloneqq \\
    &\mathsf{T}_1(ABC \to DEF) \cap \mathsf{T}_1(DEC \to ABF),
  \end{aligned}
\end{align}
while enforcing input-output inversion at the level of individual subsystems yields further constraints on the set of admissible maps, as
\begin{align}
  \label{eq:7}
  \begin{aligned}
    T_1&(\hat{A}\hat{\hat{B}}C \to \hat{D}\hat{\hat{E}}F) \coloneqq\\
    & \mathsf{T}_1(\hat{A}BC \to \hat{D}EF) \cap \mathsf{T}_1(A\hat{B}C \to D\hat{E}F).   
  \end{aligned}
 \end{align}
In the first case, if we exchange the bipartite system $AB$ with the bipartite system $DE$, we still obtain a well-defined channel. In the second case, we can exchange $A$ with $D$ and $B$ with $E$ independently and still obtain a well-defined channel. Clearly, $\mathsf{T}_1(\hat{A}\hat{\hat{B}}C \to \hat{D}\hat{\hat{E}}F) \subset \mathsf{T}_1(\widehat{AB}C \to \widehat{CD}E)$. The higher-order theory that we will develop in the following sections could easily include these multipartite generalizations. However, we decide not to do so in order to avoid cumbersome notation.

% --------------------------------------------------------------------------
% Sec. III
% --------------------------------------------------------------------------

\section{Type system}
\label{sec:types}
In this section, we lay the theoretical groundwork for studying the hierarchy of higher-order maps with indefinite temporal direction. We introduce a notion of elementary types, from which all higher-order types in the hierarchy can be defined recursively. Let us recall from the previous section that each finite-dimensional quantum system is labelled with roman letters, i.e., $A$, $B$, \dots, $Z$, while the trivial system is labelled by the letter $I$. Let us denote by $\set{A}$ the alphabet corresponding to the set of all these labels, and let us define the alphabet made of the symbols with a circumflex:
\begin{align}
  \label{eq:2}
  \hat{\set{A}} \coloneqq \{\hat{A} \mbox{ s.t. } A \in
  \set{A}, A \neq I  \} .
\end{align}
We exclude the symbol $\hat{I}$ from the set $\hat{\set{A}}$. For a given $\hat{X} \in \hat{\set{A}}$, we define
\begin{align}
  \label{eq:5}
  \set{E}_{X} \coloneqq \{\hat{X}w , w \in \set{A}^* \},
\end{align}
where the Kleene star $\set{A}^*$ is the set of all strings of symbols in $\set{A}$, including the empty string. The set $\set{E}_{X}$ is made of strings that begin with the symbol $\hat X\in\hat{\set{A}}$, followed by a string of symbols in $\set{A}$, e.g.\ $\hat{A}, \hat{A}BC , \hat{A}CDF \in \set{E}_{A}$ and $\hat{B}, \hat{B}C , \hat{B}DF \in \set{E}_{B}$. Let us now introduce the alphabet
\begin{align}
  \label{eq:8}
  \set{E} \coloneqq
  \set{A} \cup \hat{\set{A}} \cup \{ \to ,(,) \},
\end{align}
which contains:
\begin{enumerate}
    \item[(i)] all the symbols that label quantum systems (e.g., $A, B, \dots$),
    \item[(ii)] the symbols with a circumflex accent (e.g., $\hat{A}, \hat{B}$), except $\hat{I}$,
    \item[(iii)] the arrow symbol ``$\to$'' and the round brackets ``$)$'' and ``$($''.
\end{enumerate}
For a given pair of symbols $\hat{X},\hat{Y} \in \hat{\set{A}}$, we define the set $\set{B}_{X,Y} \subset \set{E}^*$ as follows:
\begin{align}\label{eq:4}
    \set{B}_{X,Y} \coloneqq \{w \in \set{E}^*, w = (\hat{X}u \to \hat{Y}v), \; u,v \in \set{A}^*\},
\end{align}
which contains strings of the form $(\hat{X}\to \hat{Y}), (\hat{X}A\to \hat{Y}), (\hat{X}ABCD \to \hat{Y} GHTU), \dots$. We now define the set of \emph{elementary types}
\begin{align} \label{eq:9}
    \begin{aligned}
        \ETy &\coloneqq \set{A}^{+} \cup \set{B}, \\
        \set{B} &\coloneqq \bigcup_{X,Y \in \set{A}} \set{B}_{X,Y} ,  
    \end{aligned}
\end{align}
with $\set{A}^+ := \set{A}^* \setminus \lambda$, where $\lambda$ is the empty string. The set $\ETy$ contains all the strings of the form $A,AB,\dots, (\hat{A} \to \hat{B}), (\hat{A}B\to \hat{C}), (\hat{A}CBD \to \hat{E} GHTU), \dots$, and so on.

Given the set of elementary types, we can recursively generate the whole hierarchy of types as follows.
\begin{definition}
    Let $\set{E}$ be the alphabet defined in Equation~(\ref{eq:8}). The set of \emph{types} of higher-order maps with time inversion is the smallest subset $\Ty \subset \set{E}^*$ such that
  \begin{align}
    \label{eq:10}
    \begin{aligned}
&       \ETy \subset \Ty, \\
&   \mbox{if } x,y \in \Ty \mbox{ then } (x\to y) \in \Ty. 
    \end{aligned}
   \end{align}
\end{definition}
A type $x \in \Ty$ is a string such as, e.g., $x =((( A\to (\hat{B}C \to \hat{D}E ) ) \to ((\hat{F} \to \hat{G})\to H))) \to (KL \to I)$.

We can notice an element of distinction if we compare the current definition of $\ETy$ with the one adopted for studying standard higher-order quantum maps~\cite{Perinotti2017, doi:10.1098/rspa.2018.0706}, where the set $\ETy$ was simply taken to be $\set A^{+}$. It follows that the set of elementary types is now strictly larger, as it contains the strings referring to partial bistochastic maps, which are therefore already cast at the ground level of the hierarchy.

\begin{definition} (Extension with an elementary type).
\label{def:extensionelementtype}
Let $x \in \Ty$ be a type and let $E \in \set{A}$ be an elementary type. The extension $x || E$ of $x$ by the elementary type $E$ is defined recursively as follows:
\begin{itemize}
\item if $x = A_1A_2\dots A_n \in \set{A}^*$ then
  $x || E \coloneqq A_1A_2\dots A_nE$;
\item if $x= \hat{X}B_1\dots B_n \to
    \hat{Y}C_1\dots C_m $
    then
     $x||E \coloneqq \hat{X}B_1\dots B_n \to
     \hat{Y}C_1\dots C_mE $;
\item if $x = y \to z$, $y,z \in \Ty$ then
     $x ||E \coloneqq y \to z||E$.
\end{itemize}
\end{definition}

\begin{definition} (Partial ordering $\preceq$). We say that type $x$ is a parent of type $y$ and we write $x \preceq_p y$ if there exists a type $z$ such that either $y = (x \rightarrow z)$ or $y = (z \rightarrow x)$. The relation $x \preceq y$ is defined as the transitive closure of the binary relation $\preceq_p$. 
\end{definition}

% --------------------------------------------------------------------------
% Sec. IV
% --------------------------------------------------------------------------

\section{Higher-order quantum maps with indefinite input-output direction}
\label{sec:axioms}

In the following section, we adopt the line of reasoning of~\cite{Perinotti2017, doi:10.1098/rspa.2018.0706} to provide a characterization of admissible events (maps) of a generic type. We begin by defining the set of admissible events of type $\ETy$. In contrast to Ref.~\cite{Perinotti2017, doi:10.1098/rspa.2018.0706}, this set encompasses not only quantum states but also bistochastic quantum transformations.

The linear sector of the framework is defined as follows:
\begin{definition} (Generalized events). If $x$ is a type in $\Ty$, the set of generalized events of type $x$, denoted by $T_{\mathbb{R}}(x)$, is defined by the following recursive definition.
\begin{itemize}
    \item for any $A_1A_2\dots A_n \in \set{A}^*$ we have $T_\mathbb{R}(A_1A_2\dots A_n) \coloneqq \mathcal{L}(\mathcal{H}_{A_1A_2\dots A_n})$.
    \item for any $x=(\hat{X}A_1\dots A_n\to \hat{Y} B_1\dots B_n) \in \set{B}$ we have $T_{\mathbb{R}}(x) \coloneqq \mathcal{L} (\mathcal{H}_{XA_1\dots A_nY B_1\dots B_n})$.
    \item if $x$ and $y$ are two types, then every Choi operator of a linear map $M: T_\mathbb{R}(x) \rightarrow T_\mathbb{R}(y)$ is a generalized event $M$ of type $(x \rightarrow y)$.
\end{itemize}
\end{definition}

From the previous definition, it is straightforward to prove the following lemma.

\begin{lemma} (Characterization of generalized events). Let
  $x$ be a type and let us consider the set of all the elementary types $A_i \neq I$ such that either $A_i$ or $\hat{A}_i$ occurs in the expression of $x$, namely
\begin{align} \label{eq:13}
    \set{A}_x  \coloneqq \{A_i \in \set{A}, A_i \neq I \mbox{ s.t. } A_i\in x  \lor \hat{A}_i \in x \}, 
\end{align}
where we define
\begin{align*}
    &    A_i\in x\quad\mbox{if }A_i\preceq x,\\
    &    \hat A_i\in x\quad\mbox{if }\exists B_i\in\set A, y,z\in\set A^*\quad\mbox{s.t. } (\hat A_iy\to\hat B_iz)\preceq x.
\end{align*}
Then $T_\mathbb{R}(x) = \mathcal{L}(\mathcal{H}_x)$, where $\mathcal{H}_x \coloneqq \otimes_{A_i \in \set{A}_x} \mathcal{H}_{A_i}$.
\end{lemma}
Let us clarify the previous definition with an example: if $x = ( A\to (\hat{B}C \to \hat{D}E ) ) \to ((\hat{F} \to \hat{G})\to H)$, then $\set{A}_x = \{A,B,C,D,E,F,G,H \}$.

The next definition will play a crucial role in defining the hierarchy of admissible maps.

\begin{definition} (Extended event). Let $x = x_1 \to x_2$ be a non-elementary type, $E \in \set{A}$, and $M \in T_\mathbb{R}(x)$. The \emph{extension} of $M$ by $E$ is the operator $M_E \in T_\mathbb{R}(x_1 || E \rightarrow x_2 || E)$, which is the Choi operator of the map $\mathcal{M} \otimes \mathcal{I}_E : T_\mathbb{R}(x || E) \rightarrow T_\mathbb{R}(y|| E)$, where $\mathcal{I}_E : \mathcal{L}(\mathcal{H}_E) \rightarrow \mathcal{L}(\mathcal{H}_E)$ is the identity map.
\end{definition}

We are now ready to introduce the notion of admissible maps. We proceed in two steps. First, the notion of an admissible map is defined for the elementary types, which include not only quantum systems but also bistochastic channels.

\begin{definition} (Elementary deterministic event). Let $A_1\dots A_n \in \set{A}^*$ and $M \in T_\mathbb{R}(A_1\dots A_n)$. We say that $M$ is a \emph{deterministic event} of type $A_1\dots A_n$ if $M \geq 0$ and $\operatorname{Tr}[M] = 1$. $\mathsf{T}_1(A_1\dots A_n)$ denotes the set of deterministic events of type $A_1\dots A_n$.

Let $x=(\hat{X}A_1\dots A_n\to \hat{Y} B_1\dots B_n) \in \set{B}$ and $M \in T_{\mathbb{R}}(x)$. We say that $M$ is a \emph{deterministic event} of type $(\hat{X}A_1\dots A_n\to \hat{Y} B_1\dots B_n)$ if $M$ is an $X \mhyphen Y$ bistochastic channel from $XA_1\dots A_n$ to $YB_1\dots B_n$. $\mathsf{T}_1(\hat{X}A_1\dots A_n\to \hat{Y} B_1\dots B_n)$ denotes the set of deterministic events of type $(\hat{X}A_1\dots A_n\to \hat{Y} B_1\dots B_n)$.
\end{definition}

\begin{definition}
\label{def:admis_elemen_aevent}
  (Elementary admissible event). We say that $M$ is an \emph{admissible event} of type $A_1\dots A_n$ if $M \geq 0$ and there exists $N \in \mathsf{T}_1(A_1\dots A_n)$ such that $M \leq N$. $T(A_1\dots A_n)$ denotes the set of admissible events of type $A_1\dots A_n$.

We say that $M$ is an \emph{admissible event} of type $(\hat{X}A_1\dots A_n\to \hat{Y} B_1\dots B_n)$ if $M$ is an $X \mhyphen Y$ bistochastic quantum operation from $XA_1\dots A_n$ to $YB_1\dots B_n$. $T(\hat{X}A_1\dots A_n\to \hat{Y} B_1\dots B_n)$ denotes the set of admissible events of type $(\hat{X}A_1\dots A_n\to \hat{Y} B_1\dots B_n)$.
\end{definition}

Once the notion of admissibility is given for the elementary types, we can define the admissible events for an arbitrary type in the hierarchy. As one should expect, the definition is recursive: if one knows what the admissible and deterministic events are for the types $x$ and $y$, the following definition tells what the admissible and deterministic events are for the type $x \to y$.

\begin{definition}\label{def:AdmEvent} (Admissible event). Let $x, y \in \Ty$ be two types, $M \in T_\mathbb{R}(x \rightarrow y)$ be an event of type $x \rightarrow y$, and $M_E \in T_\mathbb{R}(x || E \rightarrow y || E)$ be the extension of $M$ by $E$. Let $\mathcal{M}: T_\mathbb{R}(x) \rightarrow T_\mathbb{R}(y)$ and $\mathcal{M} \otimes \mathcal{I}_E : T_\mathbb{R}(x || E) \rightarrow T_\mathbb{R}(y || E)$ be the linear maps whose Choi operators are $M$ and $M_E$, respectively. Let us denote by $\mathsf{T}_1(x)$ the set of \emph{deterministic} events of type $x$ and by $T(x)$ the set of \emph{admissible} events of type $x$.

We say that $M$ is a deterministic event of type $x \to y$ if $(\mathcal{M} \otimes \mathcal{I}_E)[T(x \parallel E)] \subseteq T(y \parallel E)$ and $(\mathcal{M} \otimes \mathcal{I}_E)[\mathsf{T}_1(x \parallel E)] \subseteq \mathsf{T}_1(y \parallel E)$ for any $E \in \set{A}$.

We say that $M$ is an admissible event of type $x \to y$ if
\begin{enumerate}
    \item[(i)] $(\mathcal{M} \otimes \mathcal{I}_E)[T(x \parallel E)] \subseteq T(y \parallel E)$ for any $E \in \set{A}$, \label{item:1}
    \item[(ii)] there exists a collection $\{N_i\}_{i=1}^n \subseteq T_\mathbb{R}(x \rightarrow y)$, $0 \leq n < \infty$, such that
    \begin{enumerate}
        \item for all $1 \leq i \leq n$, the map $\mathcal{N}_i$ satisfies item~(i),
        \item $((\mathcal{M} + \sum_{i=1}^n \mathcal{N}_i) \otimes \mathcal{I}_E)[T(x || E)] \subseteq T(y || E)$,
        \item $((\mathcal{M} + \sum_{i=1}^n \mathcal{N}_i) \otimes \mathcal{I}_E)[\mathsf{T}_1(x || E)] \subseteq \mathsf{T}_1(y || E)$.
    \end{enumerate}
\end{enumerate}
\end{definition}
It is clear from Definition~\ref{def:AdmEvent} that $\mathsf{T}_1(x) \subseteq T(x)$.

We now proceed to characterize the set of admissible events for an arbitrary type, thereby covering the entire hierarchy of higher-order maps.

\begin{theorem} (Characterization of admissible events) \label{thm:char_adm_events}
Let $x$ be a type and $M \in T_\mathbb{R}(x)$. Then we have
\begin{equation}
\label{eq:20}
  M \in T(x) \Leftrightarrow M \geq 0 \wedge \exists D \in \mathsf{T}_1(x) \; \mbox{s.t.} \; M \leq D.
\end{equation}
Let $x,y$ be two types, $M \in T(x\to y)$ be an event of type $x \to y$, and let $\mathcal{M}:T_{\mathbb{R}}(x) \to T_{\mathbb{R}}(y)$ be the linear map whose Choi operator is $M$. Then we have:
  \begin{align}
    \label{eq:charadetprelimin}
    M \in \mathsf{T}_1(x \to y) \! \iff \! \! \left  \{ \!  
    \begin{array}{l}
    M \geq 0, \\
    \mathcal{M} [\mathsf{T}_1(x)] \subseteq  \mathsf{T}_1({y})
    \end{array}
\right.
  \end{align}
\end{theorem}

\begin{theorem}\label{theorem:T1}
  (Characterization of $\mathsf{T}_1(x)$). Let $x$ be a type. Then we have:
\begin{align} 
    & R \in \mathsf{T}_1(x) \Leftrightarrow R = \lambda_x \mathds{1}_x + X_x, \\ 
    & X_x \in \Delta_x \subseteq \mathsf{Trl}(\mathcal{H}_x), \quad R \geq 0,
\end{align} 
where the real positive coefficient $\lambda_x$ and the linear subspaces $\Delta_x$ are defined recursively as follows:
\begin{align}
    \lambda_E = \frac{1}{d_E}, &\quad \lambda_{\hat{X}A\to\hat{Y} B} = \frac{1}{d_Yd_B}, \quad \lambda_{x\rightarrow y} = \frac{\lambda_y}{d_x\lambda_x}, \\
    \label{eq:19} \Delta_{A_1\dots A_n} &= \mathsf{Trl}(\mathcal{H}_{A_1\dots A_n}), \\
    \nonumber \Delta_{\hat{X}A\rightarrow \hat{Y} B } &= \mathsf{Hrm}(\mathcal{H}_{XAY}) \otimes \mathsf{Trl}(\mathcal{H}_{B})\\
    \nonumber &\qquad \oplus \big( \mathsf{Trl}(\mathcal{H}_{X}) \otimes \mathsf{Hrm}(\mathcal{H}_A) \\
    \label{eq:21} &\qquad \otimes \mathsf{Trl}(\mathcal{H}_Y) \otimes \mathsf{I}_B \big), \\
    \Delta_{x\rightarrow y} &= \big(\mathsf{Hrm}(\mathcal{H}_x) \otimes \Delta_y \big)\oplus \big(\overline{\Delta}_x \otimes \Delta^\perp_y \big), \label{eq:11}
\end{align}
where, for a given Hilbert space $\mathcal{H}_A$, $\mathsf{I}_A = \operatorname{span}(\mathds{1}_A)$ is the one-dimensional subspace generated by the identity operator on $\mathcal{H}_A$, $\mathsf{Trl}(\mathcal{H}_A) = \{ X \mid \operatorname{Tr}[X] = 0, X^\dagger = X\}$ is the subspace of traceless Hermitian operators on $\mathcal{H}_A$, and $\mathsf{Hrm}(\mathcal{H}_A) = \mathsf{I}_A \oplus \mathsf{Trl}(\mathcal{H}_A)$ is the space of all Hermitian operators on $\mathcal{H}_A$. In turn, $\overline{\Delta}_x$ and $\Delta^\perp_x$ are the orthogonal complements of $\Delta_x$ in $\mathsf{Trl}(\mathcal{H}_x)$ and $\mathsf{Hrm}(\mathcal{H}_x)$, respectively.
\begin{proof}
    The proofs of Theorems~\ref{thm:char_adm_events} and~\ref{theorem:T1} closely follow those of Ref.~\cite{doi:10.1098/rspa.2018.0706}. The main difference is that the ground level of the hierarchy is larger, since it includes the partially bistochastic channels. Therefore, in order to adapt the proofs by induction of Ref.~\cite{doi:10.1098/rspa.2018.0706}, we need to verify that all the statements hold true for the elementary types $x \in \set{B}$. See Appendix~\ref{sec:proof-theorem-5.1} for details.
\end{proof}
\end{theorem}

% --------------------------------------------------------------------------
% Sec. V
% --------------------------------------------------------------------------

\section{Applications}
\label{sec:processes}

% --------------------------------------------------------------------------
\subsection{Basic higher-order operations and bistochastic supermaps}
In order to derive some relevant instances of higher-order maps for bistochastic devices, we still need to introduce the notion of a linear functional on a type $x$, expressed by the type $\overline{x} = x \rightarrow I$. This expression is particularly relevant, as it relates to the notion of parallel composition of types:
\begin{definition} \label{def:func}
    Given a type $x$, we define the type of functionals on events of type $x$ as follows:
    \begin{equation}
        \overline{x} = x \rightarrow I.
    \end{equation}
    Given two types $x$ and $y$, their tensor product is defined as
    \begin{equation}
        x \otimes y = \overline{x \rightarrow \overline{y}}.
    \end{equation}
\end{definition}

The deterministic events of types $\overline{x}$ and $x \otimes y$ are characterized as follows.
\begin{lemma}
    For types $x$ and $y$,
    \begin{align}
        \Delta_{\overline{x}} &= \overline{\Delta}_x, \quad \lambda_{\overline{x}} = \frac{1}{\lambda_x d_x}, \\
         \Delta_{x\otimes y} &= (\mathsf{I}_y \otimes \Delta_x) \oplus (\Delta_y \otimes \Delta_x) \oplus (\Delta_y \otimes \mathsf{I}_x), \\
        \lambda_{x\otimes y} &= \lambda_x \lambda_y.
    \end{align}
    \begin{proof}
        See Lemma~2 and Lemma~5 of \cite{doi:10.1098/rspa.2018.0706}, respectively.
    \end{proof}
\end{lemma}
As a relevant example, functionals on bistochastic channels, corresponding to higher-order maps transforming the latter into unit probability, are characterized by
\begin{align}
  \label{eq:17}
  \begin{aligned}
    &R \in\mathsf{T}_1((\hat{A} \rightarrow \hat{B}) \rightarrow I)  \iff \\
  &R = \frac{1}{d} \mathds{1}_{AB} + X_A \otimes \mathds{1}_B + \mathds{1}_A \otimes X_B, \\
  &R \geq 0,
  \quad X_A \in \mathsf{Trl}(\mathcal{H}_A),
  \quad X_B \in \mathsf{Trl}(\mathcal{H}_B).
  \end{aligned}
\end{align}

The following lemma shows that such functionals correspond operationally to choosing one of the two input-output directions probabilistically, preparing a state at the chosen input, and discarding the chosen output.

\begin{lemma}\label{lem:funcState}
    Let $R \in\mathsf{T}_1((\hat{A} \rightarrow \hat{B}) \rightarrow I)$. Then there exist $p \in [0,1]$ and density operators $\rho_A$ on $A$ and $\sigma_B$ on $B$ such that
    \begin{equation}
        R = p \rho_A \otimes \mathds{1}_B + (1-p) \mathds{1}_A \otimes \sigma_B.
    \end{equation}
    \begin{proof}
        See Appendix~\ref{app:funcState}.
    \end{proof}
\end{lemma}

This result shows that any deterministic functional on a bistochastic channel yields probabilities by classically selecting whether the device is used in the $A\rightarrow B$ or $B\rightarrow A$ direction \cite{Guo2024}.

Now, we assume that a bistochastic channel is transformed into a non-trivial quantum output and introduce the following basic prototype of higher-order transformations on bistochastic operations.

\begin{definition}[Bistochastic supermap]\label{def:bisupermap}
A bistochastic supermap is a higher-order map of type
\begin{equation}
   (\hat{A} \rightarrow \hat{B}) \rightarrow (P \rightarrow F),
\end{equation}
transforming a bistochastic operation into a quantum operation:
\begin{equation}
    \includegraphics[width=0.5\columnwidth]{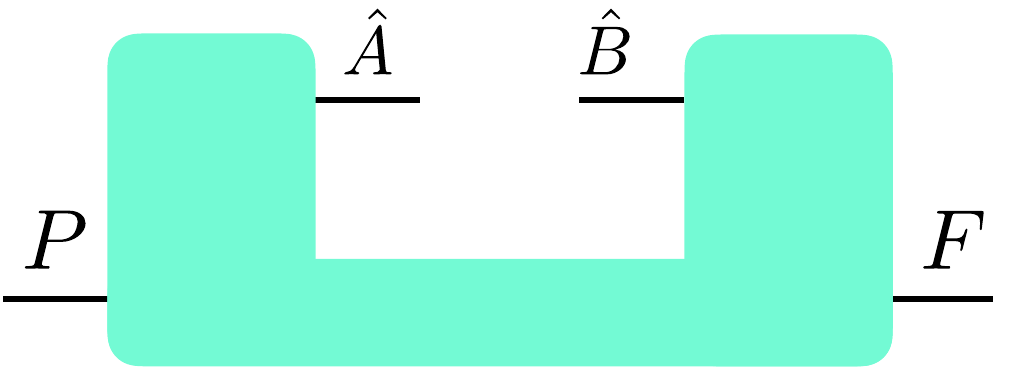}\nonumber
\end{equation}
\end{definition}

This notion provides a basic example of higher-order transformations acting on bistochastic devices and appears as a primitive in some of the constructions introduced later. A prominent example of a bistochastic supermap is the \textit{quantum time flip}~\cite{chiribella2022quantum} (see Appendix~\ref{app:QTF}), which implements a coherent control of the input-output direction of a bistochastic operation. Possible experimental implementations have been discussed in~\cite{Guo2024, stromberg2024experimental}.

% --------------------------------------------------------------------------
\subsection{Networks of higher-order maps and bistochastic generalization of quantum combs}

The notion of indefinite input-output direction captures the fact that in such a generalisation of quantum theory a physical device may not possess a predetermined orientation between its input and output systems: the same operation may be used in either direction. While a single bistochastic device already allows for such local bidirectional use, one may consider scenarios involving multiple such devices arranged in a definite order. To formalize this setting, we first recall the standard hierarchy of quantum combs, which describes higher-order maps where each inserted operation has a fixed input-output orientation and where all operations are composed in a definite causal order.

\begin{definition}\label{def:combnetwork}
For any $n\in\mathbb{N}$, an $n$-comb is a higher-order map of type
        \begin{equation}
            \mathbf{n} \coloneqq (\mathbf{n-1}) \rightarrow (P \rightarrow F), \qquad \mathbf{0} \coloneqq A \rightarrow B.
         \end{equation}
\end{definition}

The distinctive feature of this class is that every element of $\mathsf{T}_1(\mathbf{n})$ can be implemented as a quantum circuit with $n$ open slots \cite{PhysRevA.69.022309}:
\begin{equation}
    \includegraphics[width=\columnwidth]{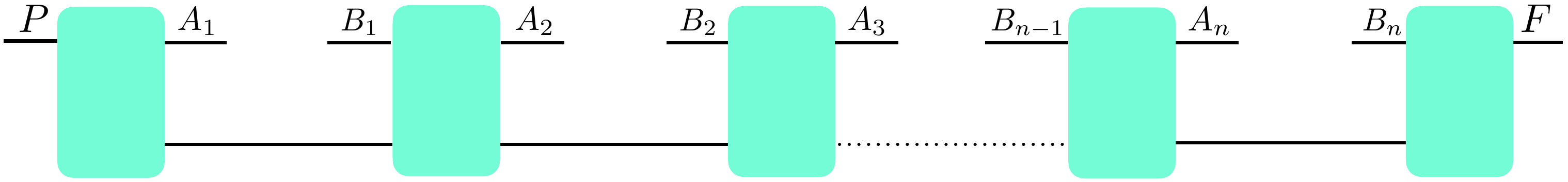}.\nonumber
\end{equation}
Each slot represents a local laboratory where the agent inserts an arbitrary quantum operation by connecting its input-output wires. This architecture enforces a fixed causal order among the operations and is known as a \emph{quantum network}.

The simplest case is the $1$-comb, a deterministic event of type $\mathbf{1} = (A_1 \rightarrow B_1) \rightarrow (P \rightarrow F)$, corresponding to a \textit{quantum supermap} transforming quantum operations into quantum operations:
\begin{equation*}
\begin{array}{l}
    \includegraphics[width=100pt]{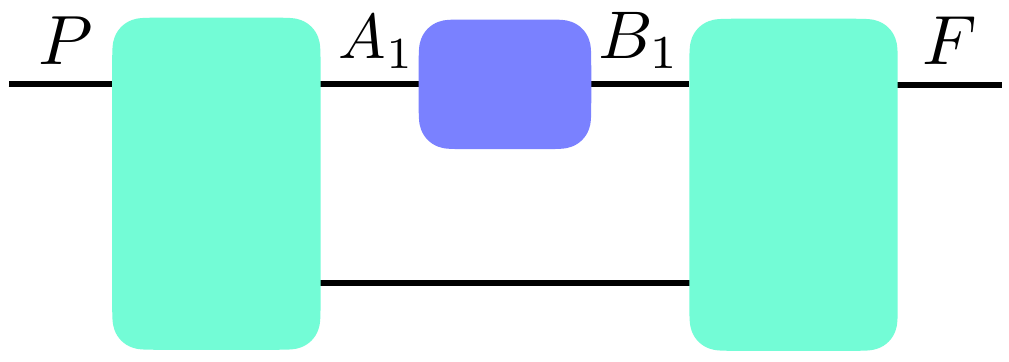}
\end{array}
=
\begin{array}{l}
\includegraphics[width=44pt]{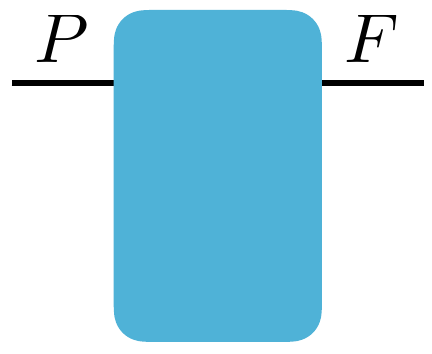}
\end{array}.
\end{equation*}

Having recalled the standard notion of quantum combs, it is natural to ask how this hierarchy can be extended when the local operations are bistochastic. Such an extension should describe local laboratories where the agents can implement operations in either input-output direction, while still preserving a well-defined global causal order. Clearly, this is a special instance of the more general scenario in which (globally) causally ordered agents can implement operations of a given type in their local laboratories. We are therefore led to consider a notion of a \emph{network of higher-order maps} (see Fig.~\ref{fig:networkhigh}).

\begin{figure}[t!]
    \centering
    \includegraphics[width=\columnwidth]{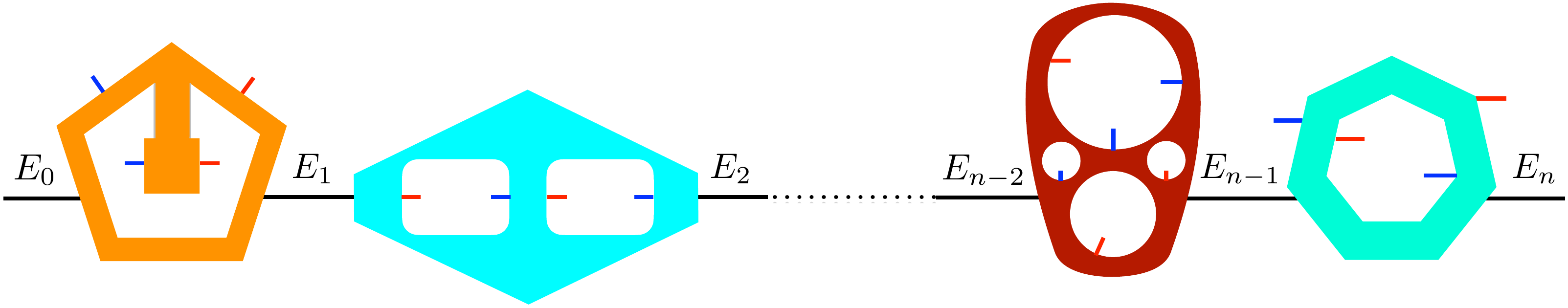}
    \caption{A network of higher-order maps allows local agents to implement more general quantum operations (namely, higher-order maps), thereby inducing an indefinite causal structure among the local quantum systems, while the corresponding local laboratories remain causally ordered with respect to one another.}
    \label{fig:networkhigh}
\end{figure}

\begin{definition}[Network of higher-order maps]
\label{def:nethigherordermaps}
  Let $n \in \mathbb{N}$, and $\{x_i\}_{i=1}^n$ be a set of types. An operator $R\in \mathcal{L}(\mathcal{H}_{E_0} \otimes \bigotimes_{i=1}^n \mathcal{H}_{x_i}\otimes  \mathcal{H}_{E_{n}})$ is a deterministic network of $n$ higher-order maps of types $\{x_i\}_{i=1}^n$ if
  \begin{align}
    \label{eq:26}
    R = M_1 * M_2 * \dots * M_n, \\
    M_i \in \mathsf{T}_1(\overline{x_i} \to (E_{i-1} \to E_{i})). \label{eq:26a}
  \end{align}
  We denote by $\mathtt{HN}_{\mathbf{n}}$ the set of deterministic networks of $n$ higher-order maps.
\end{definition}
Note that the family of networks of higher-order maps is not identified by a type within the hierarchy, but rather by a constructive definition based on the composition of higher-order maps. This raises the question of whether the composition rule introduced in Equation~\eqref{eq:26} always defines an admissible higher-order transformation. In Appendix~\ref{app:network-higher-order-welldef}, we answer this question in the affirmative.

The following result provides the mathematical characterization of the set of deterministic networks of $n$ higher-order maps.

\begin{proposition}\label{prop:higherordernetworkchar}
  An operator $R$ is a deterministic network of $n$ higher-order maps if and only if
  \begin{align}
    \label{eq:27}
   & \begin{aligned}
      R &\geq 0, \\
      R &= \lambda \mathds{1}_{E_0 x_1\ldots x_nE_n} + X, \\
    \lambda &= d^{-1}_{E_{n}}
    \prod_{i=1}^n \lambda_{x_i}, \quad X \in
    \Delta_{\mathtt{HN}_{\mathbf{n}}}, \\
    \end{aligned}
    \end{align}
    where
    \begin{align}
    \nonumber \Delta_{\mathtt{HN}_{\mathbf{n}}}& =
     \mathsf{Hrm}(\mathcal{H}_{E_0})
     \otimes \Biggl[ \bigotimes_{i=1}^n \mathsf{Hrm}(\mathcal{H}_{x_i}) \otimes \mathsf{Trl}(\mathcal{H}_{E_{n}}) \\
    & \oplus
    \bigoplus_{i=1}^{n}\Biggl(
    \bigotimes_{j=1}^{i-1}
    \mathsf{Hrm}(\mathcal{H}_{x_j}) \otimes {\Delta}_{x_i} \otimes \bigotimes_{j=i+1}^n \mathsf{I}_{x_j} \Biggr) \otimes \mathsf{I}_{E_{n}} \Biggr].
  \end{align}
\end{proposition}
\begin{proof}
  We prove the claim by induction on $n$. For $n = 1$, the statement reduces to the characterization of deterministic events of type $\overline{x_1} \to (E_0 \to E_1)$ via Theorem~\ref{theorem:T1}. Hence, the claim holds.

  \textbf{Necessity.}
  Assume the statement holds for $n-1$, and let $R$ be an operator satisfying~\eqref{eq:27} for $n$. We define:
    \begin{align}\label{eq:29}
        \nonumber S &\coloneqq R * \Bigl(\lambda_{\overline{x_n}} \mathds{1}_{x_n} \otimes \mathds{1}_{E_{n}}\Bigr) \\
        &= R * \Bigl(\frac{1}{\lambda_{x_n}d_{x_n}} \mathds{1}_{x_n} \otimes \mathds{1}_{E_{n}} \Bigr).
    \end{align}
    Using $\Tr_{x_nE_{n}}[\mathds{1}_{x_nE_{n}}] = d_{x_n}d_{E_{n}}$, we obtain
    \begin{align}\label{eq:29-1}
        S &= \alpha \mathds{1}_{E_0x_1\cdots x_{n-1}}  + Y,
    \end{align}
    where
    \begin{align}
        \nonumber \alpha &= \lambda \frac{1}{d_{x_n}\lambda_{x_n}}\Tr_{x_nE_{n}}[\mathds{1}_{x_nE_{n}}] \\
        &= \prod_{i=1}^{n-1}\lambda_{x_i}, \\
        \label{eq:defY} Y &= X*\Bigl(\frac{1}{d_{x_n}\lambda_{x_n}}\mathds{1}_{x_n}\otimes \mathds{1}_{E_{n}}\Bigr).
    \end{align}
    A straightforward calculation of \eqref{eq:defY} shows that the only non-vanishing components are ones with $\mathsf{I}_{x_n} \otimes \mathsf{I}_{E_{n}}$, proving that
    \begin{equation}
        Y \in \mathsf{Hrm}(\mathcal{H}_{E_0})
     \otimes \Biggl[\bigoplus_{i=1}^{n-1}\Biggl(
    \bigotimes_{j=1}^{i-1}
    \mathsf{Hrm}(\mathcal{H}_{x_j}) \otimes {\Delta}_{x_i} \otimes \bigotimes_{j=i+1}^{n-1} \mathsf{I}_{x_j} \Biggr) \Biggr],
    \end{equation}
    i.e., $Y \in \Delta_{\mathtt{HN}_{\mathbf{n-1}}}$ with trivial $E_{n-1}$. Therefore, $S$ obeys Equation~\eqref{eq:27} for $n-1$, and therefore it is a deterministic network of $n-1$ higher-order maps.

    Let us now define $\mathcal{H}_{E_{n-1}}\coloneqq\supp(S)$ and introduce the operator
    \begin{align}\label{eq:30}
        S' &\coloneqq \KetBra{S^{1/2*}}{S^{1/2*}},\\
        \Ket{S^{1/2*}} &\coloneqq (\mathds{1}_{E_{n-1}}\otimes S^{1/2})\Ket{\mathds{1}},
    \end{align}
    Then
    \begin{equation}
        S' = \alpha' \mathds{1} + Y' \in \mathcal{L}\Bigl( \mathcal{H}_{E_0}\otimes \bigotimes_{i=1}^{n-1}\mathcal{H}_{x_i} \otimes \mathcal{H}_{E_{n-1}} \Bigr),
    \end{equation}
    where $\alpha' > 0$ and $Y'$ is a Hermitian traceless operator. Since $\Tr[S'] = \Tr[S]$, we obtain
    \begin{equation}
        \alpha' = \frac{\alpha}{d_{E_{n-1}}}.
    \end{equation}
    We now expand $Y'$ with respect to the decomposition $\mathsf{Hrm}(E_{n-1})=\mathsf{Trl}(E_{n-1})\oplus\mathbb{C}\mathds{1}_{E_{n-1}}$:
    \begin{equation}
        Y' = \sum_k O_k\otimes A_k + W\otimes \mathds{1}_{E_{n-1}},
    \end{equation}
    where $\{A_k\}$ is a basis of traceless Hermitian operators on $\mathcal{H}_{E_{n-1}}$ and $O_k, W$ are Hermitian and Hermitian traceless operators on $\mathcal{H}_{E_0}\otimes\bigotimes_{i=1}^{n-1}\mathcal{H}_{x_i}$, respectively. Taking the partial trace over $E_{n-1}$ and using $\Tr_{E_{n-1}}[S']=S$, we obtain
    \begin{equation}
        W = Y.
    \end{equation}
    It follows that $S'$ satisfies Equation~\eqref{eq:27} for $n-1$ slots, and therefore $S'$ is a deterministic network of $n-1$ higher-order maps.

    We now define
    \begin{align}\label{eq:31}
        R' &\coloneqq (S^{-1/2}\otimes\mathds{1}_{x_nE_{n}}) R (S^{-1/2}\otimes\mathds{1}_{x_nE_{n}}) \\
        &= \beta \mathds{1} + Z \in \mathcal{L}(\mathcal{H}_{E_{n-1}}\otimes\mathcal{H}_{x_n}\otimes\mathcal{H}_{E_{n}}),
    \end{align}
    with $\beta > 0$ and $Z$ being a Hermitian traceless operator. Computing the trace yields
    \begin{equation}
        \beta = d^{-1}_{E_{n}}\lambda_{x_n}.
    \end{equation}
   Since $S$ acts trivially on $\mathcal{H}_{x_n}\otimes\mathcal{H}_{E_{n}}$, the operator $R'$ satisfies Equation~\eqref{eq:27} for $n=1$, i.e.
    \begin{equation}
        R' \in \mathsf{T}_1(\overline{x_n}\to(E_{n-1}\to E_{n})).
    \end{equation}
    Finally, using the associativity of the link product and the identity $S' * R' = R$, we conclude that $R$ is a deterministic network of $n$ higher-order maps.

   \textbf{Sufficiency.} Assume the statement holds for $n-1$, and suppose that $R$ is a deterministic network of $n$ higher-order maps. By definition, there exist deterministic maps $M_i \in \mathsf{T}_1(\overline{x_i}\to(E_{i-1}\to E_{i}))$ such that
   \begin{align}\label{eq:33}
        R &= S' * M_n, \qquad S' \coloneqq M_1 * M_2 * \dots * M_{n-1}.
    \end{align}
    Accordingly,
    \begin{align}
        S' &= \lambda' \mathds{1} + X, \qquad S' \in \mathcal{L}\Bigl( \mathcal{H}_{E_0}\otimes \bigotimes_{i=1}^{n-1}\mathcal{H}_{x_i}\otimes \mathcal{H}_{E_{n-1}}\Bigr), \\
        M_n &= \mu \mathds{1} + Y, \qquad M_n \in \mathcal{L}\Bigl(\mathcal{H}_{E_{n-1}} \otimes \mathcal{H}_{x_n} \otimes \mathcal{H}_{E_{n}} \Bigr),
    \end{align}
    where $\lambda',\mu>0$ and $X,Y$ are traceless Hermitian operators. Since $M_n \in \mathsf{T}_1(\overline{x_n} \rightarrow (E_{n-1}\rightarrow E_{n}))$, every $M_n$ is positive. The link product of positive operators is positive, hence $R\geq 0$.

    Since $M_n \in \mathsf{T}_1(\overline{x_n} \rightarrow (E_{n-1}\rightarrow E_{n}))$, it obeys Equation~\eqref{eq:27} for $n = 1$. By the inductive hypothesis, $S'$ obeys Equation~\eqref{eq:27} for $n-1$. Therefore,
    \begin{align}\label{eq:34}
        \lambda' &= d^{-1}_{E_{n-1}}\prod_{i=1}^{n-1}\lambda_{x_i}, \\
        \mu &= d^{-1}_{E_{n}}\lambda_{x_n},
    \end{align}
    and $X \in \Delta_{\mathtt{HN}_{\mathbf{n-1}}}$ and $Y \in \Delta_{\mathtt{HN}_{\mathbf{1}}}$. Expanding the link product, we obtain
   \begin{align}\label{eq:35}
        R &= (\lambda' \mathds{1}_{E_0x_1\ldots x_{n-1}E_{n-1}})*(\mu\mathds{1}_{E_{n-1} x_{n}E_{n}}) + Z \\
        &= (\lambda'\mu d_{E_{n-1}})\mathds{1}_{E_0x_1\ldots x_nE_{n}} + Z \\
        &= \Bigl(d^{-1}_{E_{n}}\prod_{i=1}^n \lambda_{x_i}\Bigr)\mathds{1}_{E_0x_1\ldots x_nE_{n}} + Z,
    \end{align}
    where
    \begin{equation}
        Z = \lambda'\mathds{1}_{E_0x_1\ldots x_{n-1}E_{n-1}}*Y + \mu\mathds{1}_{E_{n-1} x_{n}E_{n}}*X + X*Y.
    \end{equation}
    A straightforward calculation verifies that each term in $Z$ belongs to $\Delta_{\mathtt{HN}_{\mathbf{n}}}$; hence, $Z\in\Delta_{\mathtt{HN}_{\mathbf{n}}}$. Therefore, $R$ obeys Equation~\eqref{eq:27}, concluding the proof.
\end{proof}

The distinctive feature of this extended notion of network is that its local building blocks (namely, higher-order quantum maps) can represent processes with indefinite causal structure. The present characterization therefore provides a descriptive framework for the correlations that can arise in specific spacetime scenarios. In particular, it captures settings in which the causal structure may be indefinite locally, within each spacetime region or laboratory, while the causal relations between distinct regions remain definite. As a consequence, the spacetime regions themselves are constrained to follow a fixed sequential causal order. A pictorial representation of such a generalized network is shown in Fig.~\ref{fig:networkhigh}. Higher-order maps are depicted as polygonal shapes, reflecting the enlarged set of possible local interventions available to each agent, in contrast to the single open slot characterizing standard quantum combs.

A particularly relevant subclass of deterministic networks of higher-order maps is obtained by choosing each $x_i$ to correspond to a bistochastic channel.
\begin{definition}[Bi-tooth $n$-comb]\label{def:bicomb}
    Let $n\in\mathbb{N}$. An operator $R\in \mathcal{L}\!\left(\bigotimes_{i=1}^n \mathcal{H}_{A_i B_i}\right)$ is a \emph{bi-tooth $n$-comb} if it is a deterministic network of $n$ higher-order maps with $x_i = \hat{A}_{i} \rightarrow \hat{B}_{i}$ and with $E_0$ and $E_{n}$ being trivial systems. We denote by $\widehat{\mathbf{n}}_{\rm t}$ the set of bi-tooth $n$-combs.
\end{definition}

\begin{figure}[t!]
    \centering
    \includegraphics[width=\columnwidth]{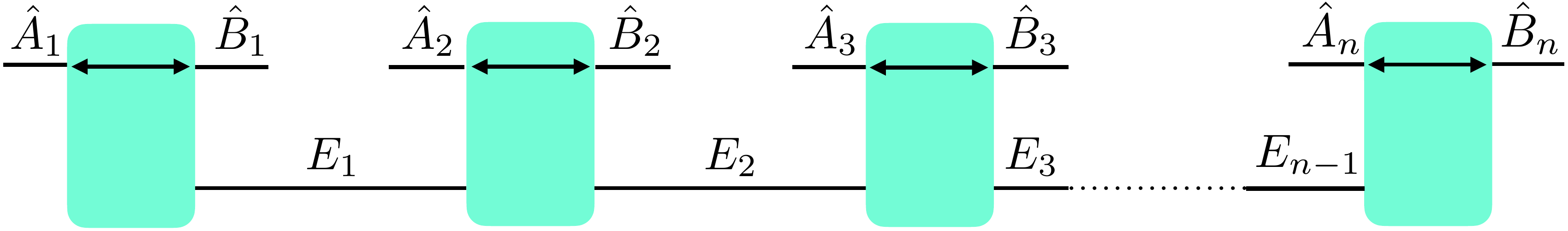}
    \caption{A bi-tooth quantum comb can be realized as a quantum circuit with open slots via sequential composition of partially bistochastic channels.}
    \label{fig:bi-tooth}
\end{figure}

Proposition~\ref{prop:higherordernetworkchar} immediately yields a complete characterization of bi-tooth $n$-combs. A positive linear operator $R \in \mathcal{L}(\mathcal{H}_{A_1B_1\ldots A_{n}B_{n}})$ is a bi-tooth $n$-comb if and only if
\begin{align}\label{eq:tCombChar}
    \begin{aligned}
        R &\geq 0, \\
        R &= \lambda_{\widehat{\mathbf{n}}_{\rm t}} \mathds{1}_{A_1B_1\ldots A_{n}B_{n}} + X, \\
        \lambda_{\widehat{\mathbf{n}}_{\rm t}} &= \frac{1}{\prod_{i=1}^{n}d_{B_i}}, \quad X \in \Delta_{\widehat{\mathbf{n}}_{\rm t}},
    \end{aligned}
\end{align}
where
\begin{align}\label{eq:deltaToothComb}
    \Delta_{\widehat{\mathbf{n}}_t} &= \bigoplus_{i=1}^{n} \Biggl[ \left(\bigotimes_{j=1}^{i-1}\mathsf{Hrm}(\mathcal{H}_{A_jB_j})\right) \nonumber \\
    & \qquad \otimes \Delta_{\hat A_i\to \hat B_i} \otimes \left(\bigotimes_{k=i+1}^{n} \mathsf{I}_{A_kB_k}\right)\Biggr].
\end{align}
Theorem~\ref{theorem:T1} suggests that $\mathsf{T}_1((\overline{\hat{A}_i \rightarrow \hat{B}_i}) \rightarrow (E_{i-1} \rightarrow E_i)) = \mathsf{T}_1(\hat{A}_i E_{i-1} \rightarrow \hat{B}_iE_i)$. Consequently, bi-tooth $n$-combs, similarly to their standard counterparts, admit an explicit realization within standard quantum circuit architectures, provided that each local transformation is allowed to be \emph{partially bistochastic}, as illustrated in Fig.~\ref{fig:bi-tooth}.

This observation naturally motivates a further generalization, namely higher-order transformations acting on bi-tooth combs and producing a quantum channel as output. Such transformations correspond to networks of higher-order maps in which each $x_i$ is taken to be a functional on bistochastic channels.

\begin{definition}[Bi-slot $n$-comb]
\label{def:biSlotComb}
    Let $n\in\mathbb{N}$. An operator $R\in \mathcal{L}(\mathcal{H}_{P} \otimes \bigotimes_{i=1}^n \mathcal{H}_{A_i B_i}\otimes \mathcal{H}_{F})$ is a \emph{bi-slot $n$-comb} if it is a deterministic network of $n$ higher-order maps with $x_i = (\hat{A}_{i} \rightarrow \hat{B}_{i} )\rightarrow I$, and with $E_0\coloneqq P$ and $E_{n} \coloneqq F$. We denote by $\widehat{\mathbf{n}}_{\rm s}$ the set of bi-slot $n$-combs.
\end{definition}

From Proposition~\ref{prop:higherordernetworkchar}, we have that $R \in \mathcal{L}(\mathcal{H}_{PA_1B_1\ldots A_{n}B_{n}F})$ belongs to $\mathsf{T}_1(\widehat{\mathbf{n}}_{\rm s})$ if and only if it is positive and
\begin{align}\label{eq:sCombChar}
    \begin{aligned}
        R &\geq 0, \\
        R &= \lambda_{\widehat{\mathbf{n}}_{\rm s}} \mathds{1}_{PA_1B_1\ldots A_{n}B_{n}F}+ X, \\
        \lambda_{\widehat{\mathbf{n}}_{\rm s}} &= \frac{1}{d_F\prod_{i=1}^{n}d_{A_i}}, \quad X \in \Delta_{\widehat{\mathbf{n}}_{\rm s}},
     \end{aligned}
\end{align}
where
\begin{align}\label{eq:deltaSlotComb}
    \nonumber \Delta_{\widehat{\mathbf{n}}_{\rm s}}& = 
     \mathsf{Hrm}(\mathcal{H}_{P})
     \otimes \Biggl[ \bigotimes_{i=1}^n \mathsf{Hrm}(\mathcal{H}_{A_i B_i}) \otimes \mathsf{Trl}(\mathcal{H}_{E_{n}}) \\
    & \oplus
    \bigoplus_{i=1}^{n}\Biggl(
    \bigotimes_{j=1}^{i-1}
    \mathsf{Hrm}(\mathcal{H}_{x_j}) \otimes \overline{\Delta}_{\hat{A}_i \rightarrow \hat{B}_i} \otimes \bigotimes_{j=i+1}^n \mathsf{I}_{x_j} \Biggr) \otimes \mathsf{I}_{F} \Biggr].
\end{align}
By direct inspection of the defining subspaces, one finds that $\widehat{\mathbf{n}}_{\rm s} = \widehat{\mathbf{n}}_{\rm t} \to (P \rightarrow F)$.

Crucially, unlike standard quantum combs, bi-slot $n$-combs do not admit a definite input-output direction for each individual slot. Nevertheless, when viewed as deterministic networks of higher-order maps, they preserve a well-defined global causal order among the local operations and admit a sequential realization analogous to that of ordinary quantum networks, now with bistochastic supermaps replacing physical channels (see Fig.~\ref{fig:realbistocomb}).

\begin{figure}[t!]
   \includegraphics[width=\columnwidth]{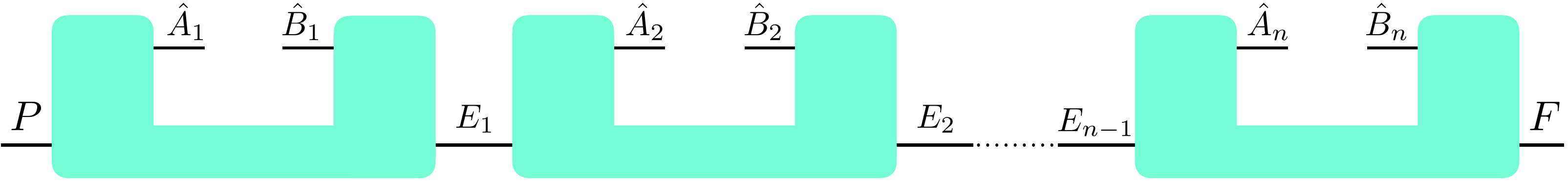}
   \caption{A bi-slot quantum comb can be realized as a sequential composition of bistochastic supermaps.}
   \label{fig:realbistocomb}
 \end{figure}

This structural characterization enables a generalization of the quantum time flip to an arbitrary number of ports, yielding the \textit{quantum $n$-time flip}. The resulting supermap preserves a well-defined global causal order among the local operations, while being incompatible with assigning a definite input-output direction to each individual slot, as detailed in Appendix~\ref{app:N-flip}.

% --------------------------------------------------------------------------
\subsection{Generalized and bistochastic process matrices}
 
The hierarchy of quantum combs characterizes higher-order maps that can be implemented as quantum circuits with causally ordered open slots. As we have seen, it extends naturally to scenarios where the local operations are higher-order maps, while retaining a global causal ordering. On the other hand, one can relax the assumption that such a global ordering exists. Higher-order transformations on local operations that do not necessarily admit a predefined causal order can be captured by the notion of \textit{process matrices}~\cite{oreshkov2012quantum}.
\begin{definition}
    For any $n\in \mathbb{N}$, an $n$-slot process matrix\footnote{In the literature, a higher-order map of type $\mathtt{P}_n$ is called ``process'', while the term ``process matrix'' commonly refers to the Choi operator associated with it. Since we work directly with Choi operators, we use the term ``process matrix'' throughout.} is a higher-order map of type
    \begin{equation}
        \mathtt{P}_n \coloneqq \Bigl(\bigotimes_{i=1}^{n} (A_i \rightarrow B_i)\Bigr) \rightarrow (P \rightarrow F), \quad \mathtt{P}_0 \coloneqq A \rightarrow B.
    \end{equation}
\end{definition}
 
For $n=1$, this definition reduces to the notion of a $1$-comb, as expected: causal order as a binary relation becomes nontrivial only when at least two operations are present. A paradigmatic example of a map of type $\mathtt{P}_2$ exhibiting indefinite causal order is the quantum \texttt{SWITCH}, which implements a coherent control of the causal order between two operations~\cite{chiribella2009beyond, chiribella2013quantum}. The construction generalizes to $n \geq 2$ operations, yielding higher-order transformations of type $\mathtt{P}_n$. The quantum \texttt{SWITCH} has been demonstrated experimentally on multiple platforms, most notably in photonic table-top experiments.

While process matrices enable indefinite causal order, they may also generate correlations between local agents that cannot arise in any scenario with fixed (or even probabilistic) causal order. Examples of such noncausal process matrices include the OCB process (type $\mathtt{P}_2$ with trivial $P$ and $F$) \cite{oreshkov2012quantum} and the AF/BW (also known as the Lugano) process (type $\mathtt{P}_3$) \cite{Baumeler_2016}.

Despite allowing indefinite causal order, standard process matrices assume that the local operation of a party is described by a channel. Similarly to what we did in Definition~\ref{def:nethigherordermaps}, we can straightforwardly generalize process matrices to the case in which a party's intervention is a generic higher-order map.
\begin{definition}
  Let $\{x_i\}_{i=1}^n$ be a set of types. A
  generalized $n$-slot process matrix is a higher-order map of type
    \begin{align}
        \mathtt{HPM}_n &\coloneqq \Bigl(\bigotimes_{i=1}^{n} x_i\Bigr) \rightarrow (P \rightarrow F), \\
        \mathtt{HPM}_0 &\coloneqq A \rightarrow B.
    \end{align}
\end{definition}

\begin{figure}[t!]
    \centering
    \includegraphics[width=110pt]{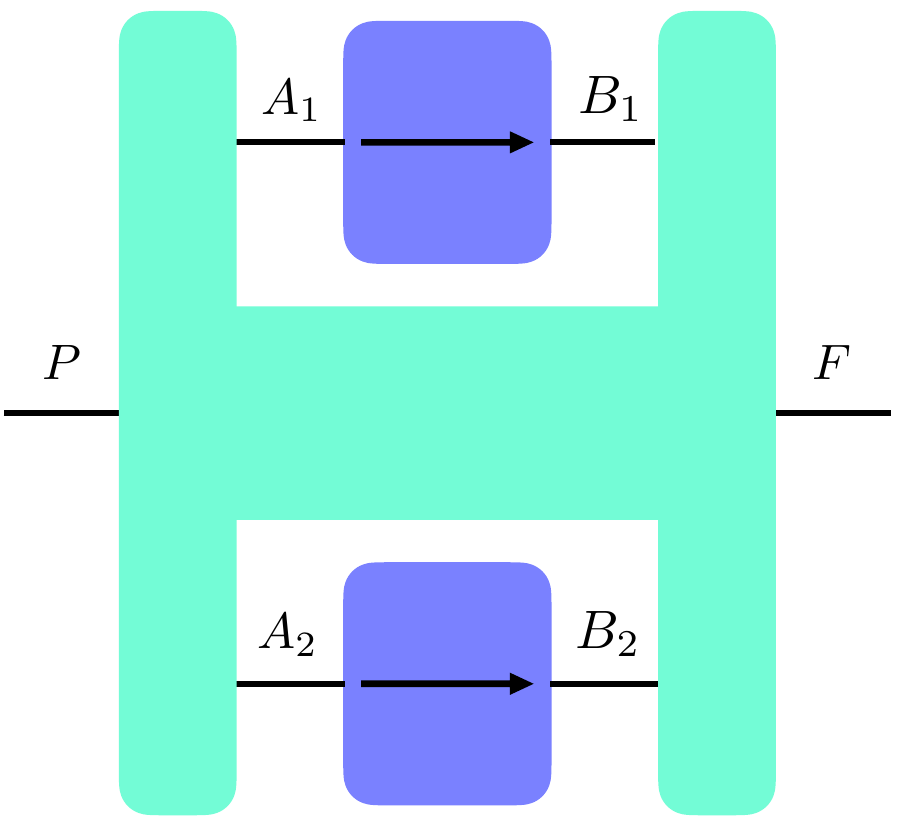}\quad\includegraphics[width=110pt]{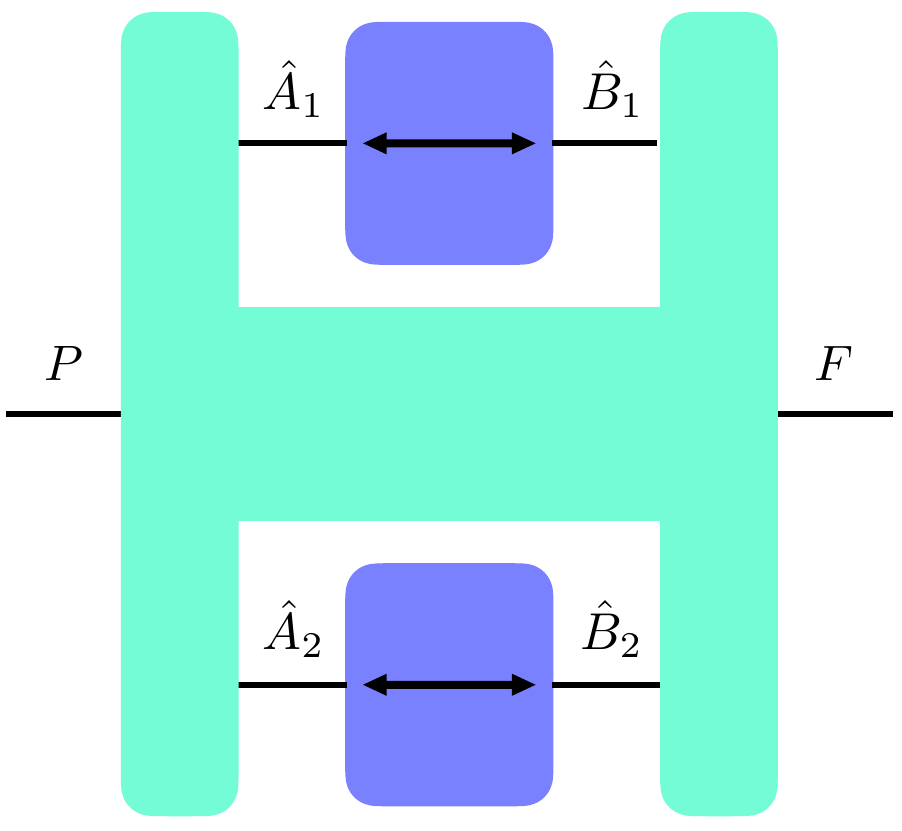}
    \caption{A standard process matrix and a bistochastic process matrix for $n=2$.}
    \label{fig:switch/biswitch}
\end{figure}

This definition encompasses scenarios in which local agents may apply bidirectional devices.

\begin{definition}
    For any $n\in \mathbb{N}$, an $n$-slot bistochastic process matrix is a higher-order map of type
    \begin{align}
        \mathtt{BSP}_n &\coloneqq \Bigl(\bigotimes_{i=1}^{n} (\hat{A}_i \rightarrow \hat{B}_i)\Bigr) \rightarrow (P \rightarrow F), \\
        \mathtt{BSP}_0 &\coloneqq A \rightarrow B.
    \end{align}
\end{definition}
It follows that a linear operator $R \in \mathcal{L}(\mathcal{H}_{PA_1B_1\ldots A_nB_nF})$ belongs to $\mathsf{T}_1(\mathtt{BSP}_n)$ (i.e., is a bistochastic process matrix) if and only if
\begin{align}\label{eq:procChar1}
  &R \geq 0, \\
  &R = \lambda_{\mathtt{BSP}_{n}} \mathds{1} + X, 
     \quad \lambda_{\mathtt{BSP}_{n}}
     = \frac{1}{d_F\prod_{i=1}^{n}d_{A_i}}, \quad X \in \Delta_{\mathtt{BSP}_{n}}, \label{eq:procChar2}
\end{align}
where
\begin{align}
            \nonumber \Delta_{\mathtt{BSP}_n} &= \mathsf{Hrm}(\mathcal{H}_P) \otimes \Biggl( \Bigl(\bigotimes_{i=1}^n (\mathsf{Hrm}(\mathcal{H}_{A_i B_i})) \Bigr)\otimes \mathsf{Trl}(\mathcal{H}_F) \\
            &\oplus \bigoplus_{i=1}^n\Biggl( \Bigl(\bigotimes_{j=1}^{i-1} (\mathsf{Hrm}(\mathcal{H}_{A_j B_j})) \Bigr) \otimes \overline{\Delta}_{\hat{A}_i \rightarrow \hat{B}_i}\nonumber\\
            & \otimes \Bigl( \bigotimes_{j'=i+1}^n (\Delta_{\hat{A}_{j'} \rightarrow \hat{B}_{j'}} \oplus \mathsf{I}_{A_j} \otimes \mathsf{I}_{B_j}) \Bigr) \Biggr) \otimes \mathsf{I}_F \Biggr),
\end{align}
where $\overline{\Delta}_{\hat{A}_i \rightarrow \hat{B}_i} = \mathsf{Trl}(\mathcal{H}_{A_i}) \otimes \mathsf{I}_{B_i} \oplus \mathsf{I}_{A_i} \otimes \mathsf{Trl}(\mathcal{H}_{B_i})$ and $\Delta_{\hat{A}_i \rightarrow \hat{B}_i} = \mathsf{Trl}(\mathcal{H}_{A_i}) \otimes \mathsf{Trl} (\mathcal{H}_{B_i})$.

Process matrices form a strict subclass of their bistochastic counterparts, indicating that bistochastic process matrices exhibit a strictly richer structure. In addition to exploiting different orders of the local operations, they also use both directions of the operations themselves (see Fig.~\ref{fig:switch/biswitch}). This enables higher-order transformations that simultaneously generalize both the quantum time flip and the quantum \texttt{SWITCH}. For instance, one can construct a bistochastic process matrix that applies two operations $\mathcal{A}$ and $\mathcal{B}$ in a superposition of the causal order $\mathcal{A}\prec\mathcal{B}$ and the reversed causal order $\mathcal{A}\succ\mathcal{B}$, where each operation is itself used in the opposite direction in the second branch (see Appendix~\ref{app:procFS}).

Moreover, bistochastic process matrices can generate correlations between local agents that exceed those achievable by ordinary process matrices. In particular, while standard process matrices obey constraints akin to Tsirelson's bound, their bistochastic counterparts can reach the algebraic maximum of such correlations \cite{Liu2025}. A concrete example of this situation is provided by a bistochastic process introduced in \cite{Liu2025}, hereafter called the ``Liu--Chiribella process'' (type $\mathtt{BSP}_2$ with trivial $P$ and $F$), which enables perfect two-way signaling (see Appendix~\ref{app:procLC}).

% --------------------------------------------------------------------------
% Sec. VI
% --------------------------------------------------------------------------

\section{Conclusions and outlook}\label{sec:conclusions}
In this work, we developed a framework for higher-order quantum theory based on bistochastic transformations, that is, transformations that preserve the maximally mixed state. The framework is grounded in an operational construction and provides a setting in which elementary devices can be applied in both temporal directions. The main goal is to formulate higher-order quantum theory without requiring a predefined assignment of input and output systems, thereby allowing bidirectional devices to be used in scenarios involving indefinite input-output direction. This is achieved through an operational notion of admissibility that leads to a recursive characterization of higher-order maps of arbitrary types involving bistochastic transformations.

From a pragmatic perspective, our framework allows for quantum information processing and computation with bidirectional devices whose input-output direction can be indefinite. Several examples of higher-order maps discussed in this work, such as the quantum time flip (already implemented experimentally in photonic setups), the flippable quantum \texttt{SWITCH}, and the Liu--Chiribella processes, can be described within our framework while lying outside ordinary higher-order quantum theory, including the standard process matrix formalism. Consequently, we introduced nontrivial bistochastic extensions of quantum combs and process matrices. In particular, by defining a hierarchy of bistochastic counterparts of quantum combs, which we call bi-slot quantum combs, we showed that although such maps are not necessarily implementable within the standard quantum circuit model, they can always be decomposed into sequences of local transformations acting on bidirectional devices. As a result, global causal order is preserved at the level of local devices, while their temporal direction is not fixed.

Several questions remain open. An important direction is the study of bidirectional operations as a resource for potential advantages in communication, computation, and channel discrimination. A particularly relevant line of research concerns quantum thermodynamics, where bistochastic transformations coincide with entropy-nondecreasing operations. Another promising direction concerns the signaling structure of higher-order maps involving bidirectional devices. In standard higher-order quantum theory, there is a fundamental correspondence between compositional structure and signaling relations \cite{Apadula2024}. However, allowing both ports of a device to act as input and output may alter this connection. This motivates further investigation into quantum causal modeling \cite{allen2017quantum}, where our framework may serve as the basis for a new class of models in which nodes do not have a fixed assignment of input and output systems.

\begin{acknowledgments}
L.A. acknowledges the support of l'Agence Nationale de la Recherche (ANR), project ANR-22-CE47-0012, the  Austrian Science Fund (FWF) [10.55776/F71] and [10.55776/COE1].

A.B. acknowledges support by European Union—Next Generation EU through the National Research Centre for HPC, Big Data and Quantum Computing, PNRR MUR Project CN0000013-ICSC. 

G.C. acknowledges support from the Chinese Ministry of Science and Technology (MOST) through grant 2023ZD0300600, from the Hong Kong Research Grant Council (RGC) through grants SRFS2021 7S02 and R7035-21F, and from the State Key Laboratory of Quantum Information Technologies and Materials, Chinese University of Hong Kong. Research at the Perimeter Institute is supported by the Government of Canada through the Department of Innovation, Science and Economic Development Canada and by the Province of Ontario through the Ministry of Research, Innovation and Science.

P.P. acknowledges financial support from European Union—Next Generation EU through the MUR project Progetti di Ricerca d’Interesse Nazionale (PRIN) QCAPP No. 2022LCEA9Y. 

K.S. thanks Jessica Bavaresco, Simon Milz, Ognyan Oreshkov, and Marco T\'ulio Quintino for the fruitful discussions and acknowledges: This research was funded in whole or in part by the Austrian Science Fund (FWF) 10.55776/PAT4559623. For open access purposes, the author applied a CC BY public copyright license to any author-accepted manuscript version arising from this submission.
\end{acknowledgments}

\bibliographystyle{apsrev4-1}
\bibliography{main}

% --------------------------------------------------------------------------
% APPENDICES
% --------------------------------------------------------------------------

\onecolumngrid
\appendix

% --------------------------------------------------------------------------
% App. A
% --------------------------------------------------------------------------

\section{Proof of Theorem \ref{thm:char_adm_events} and Theorem \ref{theorem:T1}}
\label{sec:proof-theorem-5.1}

These proofs closely follow those of Ref.~\cite{doi:10.1098/rspa.2018.0706}. The main difference is that the ground level of the hierarchy now includes partially bistochastic channels, namely the set of elementary types is now larger. If we denote by $\ETy_{q}$ the set of elementary types of higher-order quantum theory and by $\ETy_{i\text{-}o}$ the set of elementary types of higher-order quantum theory with indefinite input--output direction, we have $\ETy_{q} = \set{A}^+$, and $\ETy_{i\text{-}o} = \set{A}^+ \cup \set{B}$.

% --------------------------------------------------------------------------
\subsection{Proof of Equation \eqref{eq:20}}
Let us start with the proof of Equation \eqref{eq:20}. We will need some preliminary lemmas.

\begin{lemma} \label{lem:admEventCond}
    Let $x = y \to z$ and $X \in T_{\mathbb{R}}(x)$. Then $X \in T(x)$ if and only if it satisfies item~$(i)$ of Definition~\ref{def:AdmEvent}, and there exist $\{X_i\}_{i=1}^n \subseteq T(x)$ such that $X + \sum_{i=1}^n X_i \in \mathsf{T}_1(x)$. 
    \begin{proof}
        This statement is a restatement of Definition~\ref{def:AdmEvent}.
    \end{proof}
\end{lemma}

\begin{lemma}
    If $X,X' \in T(x)$, then $X+X' \in T(x)$ if and only if there exist $\{ X_i \} \subseteq T(x)$ such that $X + X' + \sum_i X_i \in \mathsf{T}_1(x)$.
    \begin{proof}
        The statement clearly holds for $X,X' \in T(A)$, $A \in \set{A}$. Let us now consider the case in which $X, X' \in T(\hat{A}B \to \hat{C}D)$. If $X+X' \in T(\hat{A}B \to \hat{C}D)$, then, recalling Definition~\ref{def:admis_elemen_aevent}, we have $X+X' \leq D$, where $D \in \mathsf{T}_1(\hat{A}B \to \hat{C}D)$. Then $\tilde{X} \coloneqq D - (X+X')$ satisfies $\tilde{X} \geq 0$ and $\tilde{X} \leq D$, i.e., $\tilde{X} \in T(\hat{A}B \to \hat{C}D)$. On the other hand, if there exist $\{ X_i \} \subseteq T(x)$ such that $D \coloneqq X + X' + \sum_i X_i \in \mathsf{T}_1(x)$, then $X + X' \leq D$, i.e., $X+X' \in T(\hat{A}B \to \hat{C}D)$. We can prove that the statement holds for a generic type $x \to y$ by induction, as in the proof of Lemma~A.2 of Ref.~\cite{doi:10.1098/rspa.2018.0706}.
    \end{proof}
\end{lemma}

\begin{lemma}
  Let $x = y \to z$ and $X \in T_{\mathbb{R}}(x)$. Then $X \in T(x)$ if and only if it satisfies item~\ref{item:1} of Definition~\ref{def:AdmEvent}, and there exists $X' \in T_{\mathbb{R}}(x)$ such that $X'$ satisfies item~\ref{item:1} of Definition~\ref{def:AdmEvent} and $X+X' \in \mathsf{T}_1(x)$.
\end{lemma}

\begin{corollary}\label{cor:completewithanadmissibletodeterminis}
    Let $x = y \to z$ and $X \in T_{\mathbb{R}}(x)$. Then $X \in T(x)$ if and only if it satisfies item~\ref{item:1} of Definition~\ref{def:AdmEvent}, and there exists $X' \in T(x)$ such that $X+X' \in \mathsf{T}_1(x)$.
    \begin{proof}
        The proof of these results does not involve any knowledge about the admissible elementary events. Therefore, it coincides with the proof of Lemma~A.3 of Ref.~\cite{doi:10.1098/rspa.2018.0706}.
    \end{proof}
\end{corollary}

\begin{lemma}
  Let $X \in T(x)$ and $\rho \in T(E)$, $E \in \set{A}$. Then $X \otimes \rho \in T(x \parallel E)$. Moreover, if $X \in \mathsf{T}_1(x)$ and $\rho \in \mathsf{T}_1(E)$, then $X \otimes \rho \in \mathsf{T}_1(x \parallel E)$.
 \begin{proof}
   The result holds for $x \in \set{A}^+$. Let us now consider the case in which $x \in \set{B}$, i.e., $X \in T(\hat{A}B \to \hat{C}D)$. Then $X \geq 0$, $X \leq D \in \mathsf{T}_1(\hat{A}B \to \hat{C}D)$ and $\rho \geq 0$, $\rho \leq \sigma \in \mathsf{T}_1(E)$. Recalling Definition~\ref{def:extensionelementtype}, we have $X \otimes \rho \in T_{\mathbb{R}}(\hat{A}B \to \hat{C}DE)$, $X \otimes \rho \geq 0$, and $X \otimes \rho \leq D \otimes \sigma$, which implies $X \otimes \rho \in T(\hat{A}B \to \hat{C}DE)$. Similarly, we prove that $X \in \mathsf{T}_1(x)$ and $\rho \in \mathsf{T}_1(E)$ imply $X \otimes \rho \in \mathsf{T}_1(x \parallel E)$. The statement therefore holds for $x \in \ETy$. The proof for non-elementary types is identical to that of Lemma~A.5 of Ref.~\cite{doi:10.1098/rspa.2018.0706}.
  \end{proof}
\end{lemma}

\begin{lemma}
    \label{lmm:positivecone}
    Let $x$ be a type and let us denote by $T_+(x)$ the cone $T_+(x) \coloneqq \{P \in T_{\mathbb{R}}(x) \mid \exists c \geq 0,\, P' \in T(x),\, P = cP' \}$. Then $T_+(x)$ is the cone of positive operators, namely $T_+(x) = \{ A \in \mathcal{L}(\mathcal{H}_x) \mid A \geq 0\}$. 
  \begin{proof}
    The statement holds for $x \in \set{A}$. Let us then consider $x \in \set{B}$. If $X \in T(\hat{A}B \to \hat{C}D)$, then $X \geq 0$, and therefore $T(\hat{A}B \to \hat{C}D) \subseteq P_>$, where $P_>$ denotes the cone of positive operators. On the other hand, let $A \in P_>$ be an arbitrary positive operator. Then, for sufficiently small $\epsilon > 0$, $\epsilon A \leq d_{C} d_D \mathds{1}$. Moreover, $d_{C} d_D \mathds{1} \in \mathsf{T}_1(\hat{A}B \to \hat{C}D)$. Hence, $\epsilon A \in T(\hat{A}B \to \hat{C}D)$ and $A \in T_+(x)$.

    We have thus proved that the statement holds for $x \in \ETy$. The proof for non-elementary types is identical to that of Lemma~A.6 of Ref.~\cite{doi:10.1098/rspa.2018.0706}.
  \end{proof}
\end{lemma}

Finally, we can prove that $M \in T(x)$ if and only if $M \geq 0$ and there exists $D \in \mathsf{T}_1(x)$ such that $M \leq D$. The statement holds for $x \in \ETy$. Let us consider the case in which $x = y \to z$. If $M \in T(x)$, then Corollary~\ref{cor:completewithanadmissibletodeterminis} implies that there exists $M' \in T(x)$ such that $D \coloneqq M + M' \in \mathsf{T}_1(x)$. From Lemma~\ref{lmm:positivecone} we then have that $M, M', D \geq 0$ and $M \leq D$.

On the other hand, let $M \geq 0$ and $M \leq D$ for some $D \in \mathsf{T}_1(x \parallel E)$, and suppose that the statement holds for any $y \parallel E$ with $y \preceq x$. Let us define $N \coloneqq D - M$ and let $\mathcal{M}$ and $\mathcal{N}$ be the linear maps having $M$ and $N$ as their Choi operators. Since $M, N \geq 0$, we have that $\mathcal{M}$ and $\mathcal{N}$ are completely positive. Then, from the inductive hypothesis, we have that $\mathcal{M}$ and $\mathcal{N}$ both satisfy item~$(i)$ of Definition~\ref{def:AdmEvent}, and therefore Lemma~\ref{lem:admEventCond} implies that $M, N \in T(x \parallel E)$.

% --------------------------------------------------------------------------
\subsection{Proof of Equation \eqref{eq:charadetprelimin}}

Let us now prove Equation \eqref{eq:charadetprelimin}. We will need two preliminary results.

\begin{lemma} \label{lmm:superpositive}
    For arbitrary $x$ and $y$, consider the type $x \to y$. Let $R \in \Ev{x \to y}$ be such that $R \geq 0$ and, for all $E$, $((\Ch^{-1})[R] \otimes \mathcal{I}_E)[\Evd{x \para E}] \subseteq \Evd{y \para E}$. Then $R \in \Evd{x \to y}$.
    \begin{proof}
        The proof of this result does not involve any knowledge about the admissible elementary events. Therefore, it is the same as the proof of Lemma~B.1 of Ref.~\cite{doi:10.1098/rspa.2018.0706}.
    \end{proof}
\end{lemma}

\begin{lemma}\label{lem:necesuffpart}
   For every $E,E' \in \ETy$, and for every $R \in T(x \para EE')$, one has
    \begin{align}\label{eq:partialtracedet}
        R \in \mathsf{T}_1(x \para EE') \iff \Tr_E[R] \in \mathsf{T}_1(x \para E').
    \end{align}
    \begin{proof}
        The statement is true for $x \in \mathsf{A}$. Let us now consider the case $x \in \mathsf{B}$. Fix an arbitrary $R \in T((\hat{A}B \to \hat{C}D) \para EE') = T(\hat{A}B \to \hat{C}DEE')$. First, suppose that $R \in \mathsf{T}_1(\hat{A}B \to \hat{C}DEE')$ and define $S \coloneqq \Tr_{E'}[R]$. Then $R \geq 0$, $\Tr_{CDEE'}[R] = \mathds{1}_{AB}$, and $\Tr_{ADEE'}[R] = \mathds{1}_{CB}$. Clearly, this implies that $S \geq 0$, $\Tr_{CDE}[S] = \mathds{1}_{AB}$, and $\Tr_{ADE}[S] = \mathds{1}_{CB}$, i.e., $S \in \mathsf{T}_1(\hat{A}B \to \hat{C}DE)$. On the other hand, suppose that $S \in \mathsf{T}_1(\hat{A}B \to \hat{C}DE)$, i.e., $S \geq 0$, $\Tr_{CDE}[S] = \mathds{1}_{AB}$, and $\Tr_{ADE}[S] = \mathds{1}_{CB}$. By substituting $S = \Tr_{E'}[R]$, we obtain $\Tr_{CDEE'}[R] = \mathds{1}_{AB}$ and $\Tr_{ADEE'}[R] = \mathds{1}_{CB}$. Since $R \in T(\hat{A}B \to \hat{C}DEE')$ by hypothesis, we also have $R \geq 0$, and therefore $R \in \mathsf{T}_1(\hat{A}B \to \hat{C}DEE')$. We have thus proved the statement for $x \in \ETy$. The proof for non-elementary types is identical to that of Lemma~B.2 of Ref.~\cite{doi:10.1098/rspa.2018.0706}.
    \end{proof}
\end{lemma}

We can now prove Equation \eqref{eq:charadetprelimin}. The $\Rightarrow$ direction is obvious. Let us therefore suppose that $M \geq 0$ and $\mathcal{M}[\mathsf{T}_1(x)] \subseteq \mathsf{T}_1(y)$. From Equation \eqref{eq:20} and Lemma~\ref{lmm:superpositive}, we have that $M \in \mathsf{T}_1(x \to y)$ if and only if $M \geq 0$ and $((\Ch^{-1})[M] \otimes \mathcal{I}_E)[\mathsf{T}_1(x \para E)] \subseteq \mathsf{T}_1(y \para E)$. Fix an arbitrary $X \in \mathsf{T}_1(x \para E)$ and define  $Y \coloneqq ((\Ch^{-1})[M] \otimes \mathcal{I}_E)[X]$. From Lemma~\ref{lem:necesuffpart}, we have that $Y \in \mathsf{T}_1(y \para E)$ if and only if $\Tr_E[Y] \in \mathsf{T}_1(y)$. Moreover,
\[
    \Tr_E[Y] = \Tr_E[((\Ch^{-1})[M] \otimes \mathcal{I}_E)[X]] = (\Ch^{-1})[M][\Tr_E[X]].
\]
From Lemma~\ref{lem:necesuffpart}, we have that $\Tr_E[X] \in \mathsf{T}_1(x)$, and since $\mathcal{M}[\mathsf{T}_1(x)] \subseteq \mathsf{T}_1(y)$, the claim follows.

\subsection{Proof of Theorem \ref{theorem:T1}}

We now conclude this section with the proof of Theorem~\ref{theorem:T1}. First, we have the following lemma.

\begin{lemma}\label{lmm:transposedeteffect}
    Let $x$ be a type. Then $R \in \mathsf{T}_1(x) \iff R^T \in \mathsf{T}_1(x)$.  
    \begin{proof}
        The statement is obvious for $x \in \set{A}$. It is also easy to check that $R \in \mathsf{T}_1(\hat{A}B \to \hat{C}D) \iff R^T \in \mathsf{T}_1(\hat{A}B \to \hat{C}D)$. Hence, the statement holds for any $x \in \ETy$. The proof for non-elementary types is the same as that of Lemma~5.5 of Ref.~\cite{doi:10.1098/rspa.2018.0706}. 
    \end{proof}
\end{lemma}

We can now complete the proof of Theorem~\ref{theorem:T1}. The validity of Equations \eqref{eq:19} and \eqref{eq:21} can be checked by direct inspection. Let us consider the case in which $x = y \to z$ is not elementary. From Equation~\eqref{eq:charadetprelimin}, we know that $M \in \mathsf{T}_1(y \to z)$ if and only if $M \geq 0$ and $\mathcal{M}(\mathsf{T}_1(y)) \subseteq \mathsf{T}_1(z)$, i.e., $\Tr_{y}[M (\mathds{1}_z \otimes S_y^T)] \in \mathsf{T}_1(z)$ for any $S_y \in \mathsf{T}_1(y)$. From Lemma~\eqref{eq:partialtracedet}, we have
\begin{equation}\label{eq:23}
    \forall S_y \in \mathsf{T}_1(y): \Tr_{y}[M (\mathds{1}_z \otimes S_y^T)] \in \mathsf{T}_1(z) \iff \forall S_y \in \mathsf{T}_1(y): \Tr_{y}[M (\mathds{1}_z \otimes S_y)] \in \mathsf{T}_1(z).
\end{equation}
Then Equation~\eqref{eq:11} can be proved by induction from Equation~\eqref{eq:23} (the detailed steps are the same as those in the proof of Proposition~5.6 of Ref.~\cite{doi:10.1098/rspa.2018.0706}).

% --------------------------------------------------------------------------
% App. B
% --------------------------------------------------------------------------

\section{Quantum time flip}\label{app:QTF}

\begin{definition}
    Let $\mathcal{C}\in \mathcal{L}(\mathcal{L}(\mathcal{H}_A), \mathcal{L}(\mathcal{H}_B))$ be a completely positive map with Kraus operators $\{C_i\}_i$. The quantum time flip is a supermap that transforms $\mathcal{C}$ into a new operation $\mathcal{F}(\mathcal{C}) \in \mathcal{L}(\mathcal{L}(\mathcal{H}_{P_t} \otimes \mathcal{H}_{P_c}), \mathcal{L}(\mathcal{H}_{F_t} \otimes \mathcal{H}_{F_c}))$, defined as
    \begin{align}\label{app:eq:QTF-supermap}
        \mathcal{F}(\mathcal{C})[\rho] &= \sum_i F_i (\rho\otimes\omega) F_i^\dagger, \qquad F_i \coloneqq C_i \otimes |0\rangle\langle 0| + C_i^T \otimes |1\rangle\langle 1|,
    \end{align}
    where $\mathcal{H}_{P_t} \simeq \mathcal{H}_{A} \simeq \mathcal{H}_{B} \simeq \mathcal{H}_{F_t}$ and $\mathcal{H}_{P_c} \simeq \mathcal{H}_{F_c} \simeq \mathbb{C}^2$, $\rho \in \mathcal{L}(\mathcal{H}_{P_t})$ and $\omega \in \mathcal{L}(\mathcal{H}_{P_c})$, with indices $t$ and $c$ corresponding to target and control systems, respectively.
\end{definition}

A straightforward calculation shows that the quantum time flip is defined by a Choi operator
\begin{align}
    \label{eq:QTF-app}
    R &= \dket{\mathcal{F}}\dbra{\mathcal{F}}, \\
    \nonumber \dket{\mathcal{F}} &= \dket{\mathds{1}}_{P_t A} \otimes \dket{\mathds{1}}_{B F_t} \otimes |0\rangle_{P_c}\otimes|0\rangle_{F_c} \\
    & \qquad + \dket{\mathds{1}}_{P_t B} \otimes \dket{\mathds{1}}_{A F_t} \otimes |1\rangle_{P_c}\otimes|1\rangle_{F_c},
\end{align}
and the supermap~\eqref{app:eq:QTF-supermap} can be obtained as
\begin{align}
    \nonumber \mathcal{F}(\mathcal{C})[\rho] &= R * \rho * \omega * C  \\
    &= \mathrm{Tr}_{ABP}\!\left[\dket{\mathcal{F}}\dbra{\mathcal{F}} \bigl(\mathds{1}_F\otimes \rho^T\otimes\omega^T\otimes C^T\bigr) \right], \\
    \nonumber F_i &= \dbra{C_i}^{T}\dket{\mathcal{F}} \\
    &= C_i \otimes |0\rangle\langle 0| + C_i^T \otimes |1\rangle\langle 1|,
\end{align}
for the corresponding $C \in \mathsf{T}_1(\hat A \rightarrow \hat B)$ with Kraus operators $\{C_i\}_i$.

In order to verify that $R \in \mathsf{T}_1((\hat{A} \rightarrow \hat{B}) \rightarrow (P \rightarrow F))$, we take into account that, by Theorem~\ref{theorem:T1}, this is equivalent to:
\begin{equation}
    \label{app:eq:sCond1}
    R = \lambda_R \mathds{1}_{PABF} + X_R, \qquad \lambda_R = \frac{1}{d_F d_{A}}, \qquad X_R \in \Delta_{(\hat{A} \rightarrow \hat{B}) \rightarrow (P \rightarrow F)},
\end{equation}
where
\begin{equation}
    \Delta_{(\hat{A} \rightarrow \hat{B}) \rightarrow (P \rightarrow F)} = \mathsf{Hrm}(\mathcal{H}_{PAB})\otimes \mathsf{Trl}(\mathcal{H}_F) \oplus \mathsf{Hrm}(\mathcal{H}_P)\otimes \bigl(\mathsf{Trl}(\mathcal{H}_{A})\otimes \mathsf{I}_{B} \oplus \mathsf{I}_{A}\otimes \mathsf{Trl}(\mathcal{H}_{B})\bigr) \otimes \mathsf{I}_F.
\end{equation}
Equivalently,
\begin{equation}\label{app:eq:sCond3}
    X_R \notin \mathsf{Hrm}(\mathcal{H}_P)\otimes \bigl(\mathsf{Trl}(\mathcal{H}_{A})\otimes \mathsf{Trl}(\mathcal{H}_{B}) \oplus \mathsf{I}_{A}\otimes\mathsf{I}_{B} \bigr) \otimes \mathsf{I}_F.
\end{equation}

For comparison, the corresponding ordinary supermaps are characterized by
\begin{align}
    \label{app:eq:snCond1}
    R &= \lambda_R\,\mathds{1}_{PABF} + X_R, \qquad \lambda_R = \frac{1}{d_F d_{A}},\\
    \label{app:eq:sCond2}
    X_R &\notin \mathsf{Hrm}(\mathcal{H}_P)\otimes \bigl(\mathsf{Trl}(\mathcal{H}_{A})\otimes \mathsf{Trl}(\mathcal{H}_{B}) \oplus \mathsf{I}_{A}\otimes\mathsf{Trl}(\mathcal{H}_{B}) \oplus \mathsf{I}_{A}\otimes\mathsf{I}_{B} \bigr) \otimes \mathsf{I}_F.
\end{align}
Hence, a bistochastic but non-ordinary supermap must contain traceless components in the difference
\begin{align}
    \Delta' &= \mathsf{Hrm}(\mathcal{H}_P)\otimes \mathsf{I}_{A}\otimes\mathsf{Trl}(\mathcal{H}_{B}) \otimes \mathsf{I}_{F}. \label{app:eq:diffS}
\end{align}

First, expand $R$:
\begin{align}
    R &= \dket{\mathds{1}}\!\dbra{\mathds{1}}_{P_t A_I} \otimes \dket{\mathds{1}}\!\dbra{\mathds{1}}_{A_O F_t}
       \otimes |0\rangle\langle 0|_{P_c} \otimes |0\rangle\langle 0|_{F_c} \\
    &\quad + \dket{\mathds{1}}\!\dbra{\mathds{1}}_{P_t B_O} \otimes \dket{\mathds{1}}\!\dbra{\mathds{1}}_{A_I F_t}
       \otimes |1\rangle\langle 1|_{P_c} \otimes |1\rangle\langle 1|_{F_c} \\
    &\quad + \ldots,
\end{align}
where $\ldots$ denotes the off-diagonal terms in the control subsystems $P_c$ and $F_c$. These off-diagonal terms are traceless on $F_c$ and therefore belong to the subspace
\begin{equation}
\mathsf{Hrm}(\mathcal{H}_{PAB}) \otimes \mathsf{Trl}(\mathcal{H}_F) \subsetneq \Delta_{(\hat{A} \rightarrow \hat{B}) \rightarrow (P \rightarrow F)}.
\end{equation}

Positivity of $R$ is immediate from its definition as a rank-one operator:
\begin{equation}
    R = \dket{\mathcal{F}}\!\dbra{\mathcal{F}} \;\geq\; 0.
\end{equation}
Moreover, using $\Tr(\dket{\mathds{1}}\!\dbra{\mathds{1}}_{XY}) = d_X$ for each maximally entangled pair, one checks that
\begin{equation}
    \lambda_{R}
    = \frac{\Tr[R]}{\dim(\mathcal{H}_{PABF})}
    = \frac{1}{2 d^2},
\end{equation}
in agreement with the general normalization for a bistochastic supermap.

For each pair of isomorphic systems $X,Y$, we use the standard decomposition
\begin{equation}
    \dket{\mathds{1}}\!\dbra{\mathds{1}}_{XY}
    = \frac{1}{d}\,\mathds{1}_{XY}
      + \sum_{\mu = 1}^{d^2-1} F_{\mu,X}\otimes F_{\mu,Y}^*,
\end{equation}
where $\{F_{\mu,X}\}_{\mu=1}^{d^2-1}\subset\mathsf{Trl}(\mathcal{H}_X)$ is a traceless orthonormal operator basis on $X$. Inserting this into each maximally entangled link in~\eqref{eq:QTF-app} and expanding $R$, we now focus on those contributions whose local factor on $F = F_t F_c$ is proportional to the identity operator $\mathds{1}_F$. Denote their sum by $R^{(0)}$, which is given by
\begin{align}\label{app:eq:flipDecomp}
    R^{(0)}
    &= \frac{1}{2d^2}\,\mathds{1}_{P_tABF_tP_cF_c} + \frac{1}{2d}\sum_{\mu} \Bigl(
       F_{\mu,P_t}\otimes F_{\mu,A}^*\otimes \mathds{1}_{B F_t} + 
       F_{\mu,P_t}\otimes \mathds{1}_{A} \otimes F_{\mu,B}^*\otimes \mathds{1}_{F_t}
        \Bigr) \otimes\mathds{1}_{P_c F_c}.
\end{align}
No term in $R^{(0)}$ matches \eqref{app:eq:sCond3}, so the quantum time flip belongs to the class of bistochastic supermaps. On the other hand, the last term in~\eqref{app:eq:flipDecomp} is forbidden for ordinary supermaps, as~\eqref{app:eq:diffS} suggests. Therefore, the quantum time flip cannot be implemented within the quantum circuit model.

% --------------------------------------------------------------------------
% App. C
% --------------------------------------------------------------------------

\section{Proof of Lemma \ref{lem:funcState}}
\label{app:funcState}

In what follows, we adopt the strategy used in the proof of Lemma~2 in~\cite{Guo2024}. Let $R \in \mathsf{T}_1((\hat{A} \rightarrow \hat{B}) \rightarrow I)$, so that it satisfies Eq.~\eqref{eq:17}. We decompose $R$ as
\begin{equation}
    R = R_{\rm fwd} + R_{\rm bwd} - \frac{1}{d}\mathds{1}_{AB},
\end{equation}
where
\begin{equation}
    R_{\rm fwd} \coloneqq \frac{1}{d}\mathds{1}_{AB} + X_A \otimes \mathds{1}_B, \qquad
    R_{\rm bwd} \coloneqq \frac{1}{d}\mathds{1}_{AB} + \mathds{1}_A \otimes X_B.
\end{equation}
Since $X_A\otimes\mathds{1}_B$ and $\mathds{1}_A\otimes X_B$ commute, the operators $R_{\rm fwd}$ and $R_{\rm bwd}$ commute and therefore admit a common eigenbasis $\{|u_i\rangle\}_i$, with
\begin{equation}
    R_{\rm fwd} = \sum_i \mu_i |u_i\rangle \langle u_i |, \qquad
    R_{\rm bwd} = \sum_i \tilde{\mu}_i |u_i\rangle \langle u_i |,
\end{equation}
and positivity of $R$ implies $\mu_i + \tilde{\mu}_i \geq \frac{1}{d}$ for all $i$. Let $\mu_{\rm min} \coloneqq \min_i \mu_i$. Because $\sum_i \mu_i = \operatorname{Tr}(R_{\rm fwd}) = d$ and the average eigenvalue equals $\frac{1}{d^2}\sum_i \mu_i = \frac{1}{d}$, we have $\mu_{\rm min} \leq \frac{1}{d}$. Moreover, defining
\begin{equation}
    R'_{\overline{A\rightarrow B}} = \Bigl(\frac{1}{d} - \mu_{\rm min}\Bigr)\mathds{1}_{AB} + X_A \otimes \mathds{1}_B, \qquad
    R'_{\overline{B\rightarrow A}} = \mu_{\rm min} \mathds{1}_{AB} + \mathds{1}_A \otimes X_B,
\end{equation}
the eigenvalues satisfy
\begin{equation}
    \mu_i - \mu_{\rm min} \geq 0, \qquad
    \tilde{\mu}_i - \frac{1}{d} + \mu_{\rm min} \geq 0,
\end{equation}
where the second inequality follows from $\mu_i + \tilde{\mu}_i \geq \frac{1}{d}$. Thus $R'_{\overline{A\rightarrow B}}$ and $R'_{\overline{B\rightarrow A}}$ are positive, and
\begin{equation}
    R = R'_{\overline{A\rightarrow B}} + R'_{\overline{B\rightarrow A}}.
\end{equation}
Define
\begin{equation}
    p \coloneqq d \Bigl(\frac{1}{d} - \mu_{\rm min}\Bigr) = 1 - d\mu_{\rm min},
\end{equation}
so $p \in [0,1]$. If $\mu_{\rm min} < \frac{1}{d}$, define
\begin{eqnarray}
    \nonumber R_{\overline{A\rightarrow B}} &\coloneqq& \frac{1}{p}R'_{\overline{A\rightarrow B}} \\
    &=& \frac{1}{d}\mathds{1}_{AB} + Y_A\otimes\mathds{1}_B,
\end{eqnarray}
where $Y_A \in \mathsf{Trl}(\mathcal H_A)$. By Theorem~\ref{theorem:T1}, any operator $R_{\overline{A\rightarrow B}} \in \mathsf{T}_1((A\rightarrow B)\rightarrow I)$ must be of the form
\begin{eqnarray}
    \nonumber R_{\overline{A\rightarrow B}} &=& \frac{1}{d}\mathds{1}_{AB}+Y_A\otimes\mathds{1}_B \\
    \nonumber &=&\Bigl(\frac{1}{d}\mathds{1}_A+Y_A\Bigr)\otimes\mathds{1}_B \\
    &\coloneqq&\rho_A\otimes\mathds{1}_B,
\end{eqnarray}
where $Y_A\in\mathsf{Trl}(\mathcal H_A)$ and therefore $\rho_A$ is Hermitian, positive (because $R_{\overline{A\rightarrow B}} = \rho_A\otimes \mathds{1}_B \geq 0$), and satisfies $\mathrm{Tr}[\rho_A]=1$. Thus $\rho_A$ is a density operator. By symmetry, if $\mu_{\rm min} > 0$, then $\frac{1}{1-p}R'_{\overline{B\rightarrow A}}$ is an element of $\mathsf{T}_1((B\rightarrow A)\rightarrow I)$ and must be of the form $\mathds{1}_A\otimes\sigma_B$ for some density operator $\sigma_B$. Finally,
\begin{equation}
    R = p(\rho_A \otimes \mathds{1}_B) + (1-p)(\mathds{1}_A \otimes \sigma_B),
\end{equation}
proving the thesis.

% --------------------------------------------------------------------------
% App. D
% --------------------------------------------------------------------------

\section{Networks of higher-order maps are well-defined}
\label{app:network-higher-order-welldef}
In this appendix, we prove that a network of higher-order maps given by Definition~\ref{def:nethigherordermaps} corresponds to an admissible event. For this purpose, we recall the notion of admissible composition from Ref.~\cite{Apadula2024}.

\begin{definition}
  \label{def:admiscomp}
Let $x$ and $y$ be two types such that they share a set of elementary types. We say that the composition of $x$ and $y$ is \emph{admissible} if
\begin{align}
  \begin{split}
  \label{eq:7-1}
  \forall R \in \mathsf{T}_1(x), \, \forall S \in \mathsf{T}_1(y), \; \exists
  z \ \mbox{s.t.}\ R * S \in \mathsf{T}_1(z).
  \end{split}
\end{align}
\end{definition}

First, we prove the following result.
\begin{lemma}\label{app:lem:composition}
Let $x$ and $y$ be arbitrary types, and let $R \in \mathsf{T}_1(x\to(A\to B))$ and $S \in \mathsf{T}_1(y \to (B\to C))$. Then
\begin{align}
    R * S \in \mathsf{T}_1((x\otimes y) \to (A \to C)).
\end{align}
\end{lemma}

\begin{proof}
The proof is a direct computation. From Theorem~\ref{theorem:T1} we have
\begin{align}
    \Delta_{x\to(A\to B)} &= 
    \mathsf{Hrm}(\mathcal{H}_{x}) \otimes \mathsf{Hrm}(\mathcal{H}_{A})
    \otimes \mathsf{Trl}(\mathcal{H}_{B})
    \oplus
    \overline{\Delta}_{x} \otimes \mathsf{Hrm}(\mathcal{H}_{A})
    \otimes \mathsf{I}_B, \\
    \overline{\Delta}_{x \otimes y} &= 
    \overline{\Delta}_{x}\otimes (\mathsf{I}_y \oplus \Delta_{y}) 
    \oplus
    \mathsf{Hrm}(\mathcal{H}_{x}) \otimes
    \overline{\Delta}_{y}.
\end{align}
Then, for any $R \in \mathsf{T}_1(x\to (A\to B))$ and 
$S \in \mathsf{T}_1(y\to (B\to C))$, we have 
\begin{align}
    &R = \lambda \mathds{1}_x \otimes \mathds{1}_{AB}
      + \sum_i O_i \otimes Q_i
      + \sum_{m} D_{m} \otimes K_{m} \otimes \mathds{1}_{B}, \\
    &S = \lambda' \mathds{1}_y \otimes \mathds{1}_{BC}
      + \sum_j O'_j \otimes Q'_j
      + \sum_n D_n' \otimes K'_n \otimes \mathds{1}_{C},
\end{align}
where $\lambda = \frac{1}{d_x d_B \lambda_x}$, $\lambda' = \frac{1}{d_y d_C \lambda_y}$, and
\begin{align}
    &O_i \in \mathsf{Hrm}(\mathcal{H}_{x}), \quad
    Q_i \in \mathsf{Hrm}(\mathcal{H}_{A}) \otimes \mathsf{Trl}(\mathcal{H}_{B}),
    \quad
    D_m \in \overline{\Delta}_x, \quad K_m \in \mathsf{Hrm}(\mathcal{H}_A), \\
    &O'_j \in \mathsf{Hrm}(\mathcal{H}_{y}), \quad
    Q'_j \in \mathsf{Hrm}(\mathcal{H}_{B}) \otimes \mathsf{Trl}(\mathcal{H}_{C}),
    \quad
    D'_n \in \overline{\Delta}_y, \quad K'_n \in \mathsf{Hrm}(\mathcal{H}_B).
\end{align}
A direct evaluation of the link product yields
\begin{align}
R * S
&=
\lambda \lambda' d_B \,
\mathds{1}_x \otimes \mathds{1}_y \otimes \mathds{1}_{AC}
\nonumber\\
&\quad
+ \sum_{i,j} O_i \otimes O'_j \otimes (Q_i * Q'_j)
+ \sum_j \mathds{1}_x \otimes O'_j \otimes \mathds{1}_A \otimes \operatorname{Tr}_B[Q'_j]
+ \sum_{m,j} D_m \otimes O'_j \otimes K_m \otimes \operatorname{Tr}_B[Q'_j]
\nonumber\\
&\quad
+ \Bigl(
\lambda \sum_n \operatorname{Tr}[K'_n]\,
\mathds{1}_x \otimes D'_n \otimes \mathds{1}_A
+ \lambda' d_B \sum_m D_m \otimes \mathds{1}_y \otimes K_m
\nonumber\\
&\qquad\quad
+ \sum_{m,n} \operatorname{Tr}[K'_n]\,
D_m \otimes D'_n \otimes K_m
+ \sum_{i,n} O_i \otimes D'_n \otimes (Q_i * K'_n)
\Bigr) \otimes \mathds{1}_C .
\label{app:eq:tensor_from_star}
\end{align}
By Theorem~\ref{theorem:T1}, an operator $W \in \mathcal{L}(\mathcal{H}_x \otimes \mathcal{H}_y \otimes \mathcal{H}_{AC})$ belongs to
$\mathsf{T}_1((x \otimes y) \to (A \to C))$ if and only if it is positive and can be written as
\[
W = \tilde{\lambda}\,
\mathds{1}_x \otimes \mathds{1}_y \otimes \mathds{1}_{AC}
+ X,
\qquad
\tilde{\lambda} = \frac{1}{d_x d_y d_C \lambda_x \lambda_y}, \qquad X \in \Delta_{(x \otimes y) \to (A \to C)},
\]
where
\begin{equation}
    \Delta_{(x \otimes y) \to (A \to C)} =
    \mathsf{Hrm}(\mathcal{H}_x) \otimes \mathsf{Hrm}(\mathcal{H}_y) \otimes \mathsf{Hrm}(\mathcal{H}_A) \otimes \mathsf{Trl}(\mathcal{H}_C)
    \;\oplus\;
    \overline{\Delta}_{x \otimes y} \otimes \mathsf{Hrm}(\mathcal{H}_A) \otimes \mathsf{I}_C .
\label{app:eq:tensor_map}
\end{equation}
Comparing it with \eqref{app:eq:tensor_from_star}, we conclude that
\[
R * S \in \mathsf{T}_1\!\bigl((x \otimes y) \to (A \to C)\bigr),
\]
which proves the claim.
\end{proof}

Recalling Equations~\eqref{eq:26} and \eqref{eq:26a}, we can now show that a network of higher-order maps defines an admissible composition, and therefore an admissible event. Indeed, each link product appearing in Equation~\eqref{eq:26} involves operators
\[
R \in \mathsf{T}_1(\overline{x_i} \to (E_{i-1} \to E_i)),
\qquad
S \in \mathsf{T}_1(\overline{x_{i+1}} \to (E_i \to E_{i+1})),
\]
for $i \in \{1,\ldots,n-1\}$. By Lemma~\ref{app:lem:composition}, each such link product is an admissible composition. Applying Lemma~\ref{app:lem:composition} iteratively to all links in the network therefore proves that the entire composition is admissible.

% --------------------------------------------------------------------------
% App. E
% --------------------------------------------------------------------------

\section{Quantum $n$-time flip}\label{app:N-flip}

A bi-slot $n$-comb can be realized as a sequential
composition of $n$ bistochastic supermaps. As a special case of this construction, we can
consider the generalization of the quantum time flip from a single bistochastic operation to an arbitrary number $n$ of them, implemented coherently in
sequence. Concretely, we define
\begin{equation}
    R_n = \bigast_{k=1}^n \dket{\mathcal{F}_k}\dbra{\mathcal{F}_k},
\end{equation}
where each $\dket{\mathcal{F}_k}\dbra{\mathcal{F}_k} \in \mathsf{T}_1\bigl((\hat{A}_k \rightarrow \hat{B}_k) \rightarrow (E_{k-1} \rightarrow E_k)\bigr)$ is defined analogously to~\eqref{eq:QTF-app}, with $E_0 = P$ and $E_{n} = F$. A straightforward calculation yields the following definition.

\begin{definition}
     Let $\mathcal{C}_k \in \mathcal{L}(\mathcal{L}(\mathcal{H}_{A_k}), \mathcal{L}(\mathcal{H}_{B_k}))$ for $k \in \{1,\ldots,n\}$ be completely positive maps with Kraus operators $\{C^{(k)}_i\}_i$. The quantum $n$-time flip is a supermap transforming the set $\{\mathcal{C}_k\}_k$ into an operation $\mathcal{F}^{(n)} \in \mathcal{L}(\mathcal{L}(\mathcal{H}_{P_t} \otimes \mathcal{H}_{P_c}) , \mathcal{L}(\mathcal{H}_{F_t} \otimes \mathcal{H}_{F_c}))$, defined as follows:
    \begin{align}\label{app:eq:nQTF-supermap}
        \mathcal{F}^{(n)}(\mathcal{C})[\rho] &= \sum_{i_1,\ldots,i_n} F_{i_1\ldots i_n} (\rho\otimes\omega) F_{i_1\ldots i_n}^\dagger, \\
        F_{i_1\ldots i_n} &\coloneqq \prod_{k=1}^n \Bigl( C_{i_k}^{(k)} \otimes |0\rangle\langle 0|_{P_{c^{(k)}}} + C_{i_k}^{(k)T} \otimes |1\rangle\langle 1|_{P_{c^{(k)}}} \Bigr) \nonumber\\
        &= \sum_{\mathbf{b}} \mathcal{T}_{\mathbf{b}}\Bigl(\prod_{k=1}^n  C_{i_k}^{(k)}\Bigr) \otimes |\mathbf{b}\rangle\langle\mathbf{b}|_{P_c},
    \end{align}
    where $\mathcal{H}_{P_t} \simeq \mathcal{H}_{A} \simeq \mathcal{H}_{B} \simeq \mathcal{H}_{F_t}$ and $\mathcal{H}_{P_c} \simeq \mathcal{H}_{F_c} \simeq \mathbb{C}^{2^n}$, with $\mathcal{H}_{P_c} = \otimes_{k=1}^n \mathcal{H}_{P_{c^{(k)}}}$ and $\mathcal{H}_{F_c} = \otimes_{k=1}^n \mathcal{H}_{F_{c^{(k)}}}$, $\mathcal{H}_{P_{c^{(k)}}} \simeq \mathcal{H}_{F_{c^{(k)}}} \simeq \mathbb{C}^2$ for every $k \in \{1, \ldots, n\}$, $\rho \in \mathcal{L}(\mathcal{H}_{P_t})$, and $\omega \in \mathcal{L}(\mathcal{H}_{P_c})$, with indices $t$ and $c$ corresponding to target and control systems, respectively. Moreover, $\mathbf{b}$ is an $n$-bit string encoding the state of the overall control system $\mathcal{H}_{P_c}$, and $\mathcal{T}_{\mathbf{b}}$ applies a transpose to the $k$-th operator if and only if the $k$-th bit of $\mathbf{b}$ equals $1$.
\end{definition}
Thus, the quantum $n$-time flip is a higher-order operation that coherently controls the input-output direction for each of the $n$ bidirectional devices (rendering each locally indefinite) while preserving a definite global causal order.

% --------------------------------------------------------------------------
% App. F
% --------------------------------------------------------------------------

\section{Bistochastic process matrices}
\label{app:BSP}

In this appendix, we provide several bipartite examples of bistochastic process matrices that cannot arise as ordinary process matrices. Therefore, we focus on deterministic higher-order operations of type $\mathtt{BSP}_2$. By Theorem~\ref{theorem:T1}, an operator $R \in \mathcal{L}(\mathcal{H}_{PA_1B_1A_2B_2F})$ is a deterministic event of this type if and only if
\begin{equation}
    \label{app:eq:lcCond1}
    R = \lambda_R \mathds{1}_{PA_1B_1A_2B_2F} + X_R, \qquad \lambda_R = \frac{1}{d_F d_{A_1} d_{A_2}}, \qquad X_R \in \Delta_{\mathtt{BSP}_2},
\end{equation}
where
\begin{align}
    \Delta_{\mathtt{BSP}_2} &= \mathsf{Hrm}(\mathcal{H}_{PA_1B_1A_2B_2}) \otimes \mathsf{Trl}(\mathcal{H}_F) \nonumber\\
    &\qquad \oplus \mathsf{Hrm}(\mathcal{H}_P) \otimes \Bigl(\mathsf{Hrm}(\mathcal{H}_{A_1B_1}) \otimes \bigl(\mathsf{Trl}(\mathcal{H}_{A_2}) \otimes \mathsf{I}_{B_2} \oplus \mathsf{I}_{A_2} \otimes \mathsf{Trl}(\mathcal{H}_{B_2})\bigr) \nonumber\\
    &\qquad \oplus \bigl(\mathsf{Trl}(\mathcal{H}_{A_1}) \otimes \mathsf{I}_{B_1} \oplus \mathsf{I}_{A_1} \otimes \mathsf{Trl}(\mathcal{H}_{B_1})\bigr) \otimes \bigl(\mathsf{Trl}(\mathcal{H}_{A_2}) \otimes \mathsf{Trl}(\mathcal{H}_{B_2}) \oplus \mathsf{I}_{A_2} \otimes \mathsf{I}_{B_2}\bigr)\Bigr) \otimes \mathsf{I}_F.
\end{align}
Equivalently,
\begin{align} 
    X_R &\notin \mathsf{Hrm}(\mathcal{H}_P) \otimes \Bigl(\mathsf{I}_{A_1B_1} \otimes \mathsf{Trl}(\mathcal{H}_{A_2})\otimes \mathsf{Trl}(\mathcal{H}_{B_2}) 
    \oplus \mathsf{Trl}(\mathcal{H}_{A_1})\otimes \mathsf{Trl}(\mathcal{H}_{B_1})\otimes \mathsf{I}_{A_2B_2}\nonumber\\
    &\qquad \oplus \mathsf{Trl}(\mathcal{H}_{A_1})\otimes \mathsf{Trl}(\mathcal{H}_{B_1})\otimes \mathsf{Trl}(\mathcal{H}_{A_2})\otimes \mathsf{Trl}(\mathcal{H}_{B_2}) \Bigr) \otimes \mathsf{I}_F. \label{app:eq:lcCond3}
\end{align}
For comparison, the corresponding ordinary bipartite process matrices of type $\mathtt{P}_2$ are characterized by
\begin{align}
    \label{app:eq:lcPCond1}
    R &= \lambda_R\,\mathds{1}_{PA_1B_1A_2B_2F} + X_R, \qquad \lambda_R = \frac{1}{d_F d_{A_1} d_{A_2}},\\
    \label{app:eq:lcPCond2}
    X_R &\notin \mathsf{Hrm}(\mathcal{H}_P) \otimes \Bigl(\mathsf{I}_{A_1B_1} \otimes \mathsf{Hrm}(\mathcal{H}_{A_2})\otimes \mathsf{Trl}(\mathcal{H}_{B_2})
    \oplus \mathsf{Hrm}(\mathcal{H}_{A_1})\otimes \mathsf{Trl}(\mathcal{H}_{B_1})\otimes \mathsf{I}_{A_2B_2} \nonumber\\
    &\qquad \oplus \mathsf{Hrm}(\mathcal{H}_{A_1})\otimes \mathsf{Trl}(\mathcal{H}_{B_1})\otimes \mathsf{Hrm}(\mathcal{H}_{A_2})\otimes \mathsf{Trl}(\mathcal{H}_{B_2})\Bigr) \otimes \mathsf{I}_{F}.
\end{align}
Hence, a bistochastic but non-ordinary process matrix must contain traceless components in the difference
\begin{align}
\Delta_{\mathtt{BSP}_2 \setminus \mathtt{P}_2} 
&= \mathsf{Hrm}(\mathcal{H}_{P}) \otimes \Bigl(\mathsf{I}_{A_1B_1A_2}\otimes \mathsf{Trl}(\mathcal{H}_{B_2})
\oplus \mathsf{I}_{A_1}\otimes \mathsf{Trl}(\mathcal{H}_{B_1})\otimes \mathsf{I}_{A_2B_2}\nonumber\\
&\quad \oplus \mathsf{I}_{A_1}\otimes \mathsf{Trl}(\mathcal{H}_{B_1})\otimes \mathsf{I}_{A_2}\otimes \mathsf{Trl}(\mathcal{H}_{B_2})\nonumber\\
&\quad \oplus \mathsf{Trl}(\mathcal{H}_{A_1})\otimes \mathsf{Trl}(\mathcal{H}_{B_1})\otimes \mathsf{I}_{A_2}\otimes \mathsf{Trl}(\mathcal{H}_{B_2})\nonumber\\
&\quad \oplus \mathsf{I}_{A_1}\otimes \mathsf{Trl}(\mathcal{H}_{B_1})\otimes \mathsf{Trl}(\mathcal{H}_{A_2})\otimes \mathsf{Trl}(\mathcal{H}_{B_2}) \Bigr) \otimes \mathsf{I}_{F}. \label{app:eq:diffBSP}
\end{align}

% --------------------------------------------------------------------------
\subsection{Flippable quantum \texttt{SWITCH}}
\label{app:procFS}

\begin{definition}
    Let $\mathcal{A} \in \mathcal{L}(\mathcal{L}(\mathcal{H}_{A_1}), \mathcal{L}(\mathcal{H}_{B_1}))$ and $\mathcal{B} \in \mathcal{L}(\mathcal{L}(\mathcal{H}_{A_2}), \mathcal{L}(\mathcal{H}_{B_2}))$ be completely positive maps with Kraus operators $\{A_i\}_i$ and $\{B_j\}_j$, respectively. The flippable quantum \texttt{SWITCH} is a supermap transforming the set $\{\mathcal{A}, \mathcal{B}\}$ into an operation $\mathcal{FS}(\mathcal{A}, \mathcal{B}) \in \mathcal{L}(\mathcal{L}(\mathcal{H}_{P_t} \otimes \mathcal{H}_{P_c}) , \mathcal{L}(\mathcal{H}_{F_t} \otimes \mathcal{H}_{F_c}))$ defined as follows:
    \begin{eqnarray}
        \mathcal{FS}(\mathcal{A}, \mathcal{B})[\rho] &=& \sum_{ij} K_{ij} (\rho \otimes \omega) K_{ij}^\dagger, \\
        K_{ij} &=& B_j A_i \otimes |0\rangle\langle 0| + (B_j A_i)^T \otimes |1\rangle\langle 1|,
    \end{eqnarray}
     where $\mathcal{H}_{P_t} \simeq \mathcal{H}_{A_1} \simeq \mathcal{H}_{B_1} \simeq \mathcal{H}_{A_2} \simeq \mathcal{H}_{B_2} \simeq \mathcal{H}_{F_t}$ and $\mathcal{H}_{P_c} \simeq \mathcal{H}_{F_c} \simeq \mathbb{C}^2$, $\rho \in \mathcal{L}(\mathcal{H}_{P_t})$ and $\omega \in \mathcal{L}(\mathcal{H}_{P_c})$, with indices $t$ and $c$ corresponding to target and control systems, respectively.
\end{definition}

A straightforward calculation shows that the flippable quantum \texttt{SWITCH} is defined by the Choi operator
\begin{align}
    \label{eq:QTF-app-Choi}
    R &= \dket{\mathcal{FS}}\dbra{\mathcal{FS}}, \\
    \nonumber \dket{\mathcal{FS}} &= |\mathds{1}\rangle\!\rangle_{P_t A_1} \otimes |\mathds{1}\rangle\!\rangle_{B_1 A_2} \otimes |\mathds{1}\rangle\!\rangle_{B_2 F_t} \otimes |0\rangle_{P_c} \otimes |0\rangle_{F_c} \\
        &\qquad + |\mathds{1}\rangle\!\rangle_{P_t B_2} \otimes |\mathds{1}\rangle\!\rangle_{A_2 B_1} \otimes |\mathds{1}\rangle\!\rangle_{A_1 F_t}  \otimes |1\rangle_{P_c} \otimes |1\rangle_{F_c},
\end{align}
which can be expanded as follows:
\begin{align}
    R &= |\mathcal{FS}\rangle\!\rangle \langle\!\langle\mathcal{FS}| \\
    &= |\mathds{1}\rangle\!\rangle\langle\!\langle\mathds{1}|_{P_t A_1} \otimes |\mathds{1}\rangle\!\rangle\langle\!\langle\mathds{1}|_{B_1 A_2} \otimes |\mathds{1}\rangle\!\rangle\langle\!\langle\mathds{1}|_{B_2 F_t} \otimes |0\rangle\langle 0|_{P_c} \otimes |0\rangle\langle 0|_{F_c} \\
    &\quad + |\mathds{1}\rangle\!\rangle\langle\!\langle\mathds{1}|_{P_t B_2} \otimes |\mathds{1}\rangle\!\rangle\langle\!\langle\mathds{1}|_{A_2 B_1} \otimes |\mathds{1}\rangle\!\rangle\langle\!\langle\mathds{1}|_{A_1 F_t} \otimes |1\rangle\langle 1|_{P_c} \otimes |1\rangle\langle 1|_{F_c} \\
    &\quad + \ldots,
\end{align}
where $\ldots$ denotes the off-diagonal terms in the control subsystems $P_c$ and $F_c$. These off-diagonal terms are traceless on $F_c$ and therefore belong to the subspace
\begin{equation}
    \mathsf{Hrm}(\mathcal{H}_{PA_1B_1A_2B_2}) \otimes \mathsf{Trl}(\mathcal{H}_F) \subsetneq \Delta_{\mathtt{BSP}_2}.
\end{equation}
Positivity of $R$ is immediate from its definition as a rank-one operator,
\begin{equation}
    R = |\mathcal{FS}\rangle\!\rangle \langle\!\langle\mathcal{FS}| \geq 0.
\end{equation}
Moreover,
\begin{equation}
    \lambda_R = \frac{\Tr[R]}{\dim(\mathcal{H}_{PA_1B_1A_2B_2F})} = \frac{\langle\!\langle\mathcal{FS}|\mathcal{FS}\rangle\!\rangle}{\dim(\mathcal{H}_{PA_1B_1A_2B_2F})} = \frac{1}{2d^3},
\end{equation}
where the identity $\mathrm{Tr}(|\mathds{1}\rangle\!\rangle\langle\!\langle\mathds{1}|_{XY}) = d_X$ is used for each maximally entangled pair. Hence, $\lambda_R$ agrees with the general form~\eqref{app:eq:lcCond1} for the given local dimensions.

To relate $R$ to the structure of $\Delta_{\mathtt{BSP}_2}$, we use the standard decomposition
\begin{equation}
|\mathds{1}\rangle\!\rangle\langle\!\langle\mathds{1}|_{XY} = \frac{1}{d} \mathds{1}_{XY} + \sum_{\mu = 1}^{d^2-1} F_{\mu, X} \otimes F^*_{\mu, Y},
\end{equation}
where each $F_{\mu, X} \in \mathsf{Trl}(\mathcal{H}_X)$ is traceless. Inserting this decomposition into each maximally entangled link and expanding $R$, we now focus on those contributions $R^{(0)}$ whose local factor on $F = F_t F_c$ is proportional to the identity operator $\mathds{1}_F$. The result is
\begin{align}
    R^{(0)} &= \frac{1}{2d^3} \mathds{1}_{P_tA_1B_1A_2B_2F_tP_cF_c} \nonumber\\
    &\quad + \frac{1}{2d^2}\sum_{\mu} \Bigl( F_{\mu, P_t} \otimes F_{\mu, A_1}^* \otimes \mathds{1}_{B_1A_2B_2F_t} + 2\mathds{1}_{P_tA_1} \otimes F_{\mu, B_1} \otimes F_{\mu, A_2}^* \otimes \mathds{1}_{B_2F_t} \nonumber\\
    &\qquad\qquad + F_{\mu, P_t} \otimes \mathds{1}_{A_1B_1A_2} \otimes F_{\mu, B_2}^* \otimes \mathds{1}_{F_t} \Bigr) \otimes \mathds{1}_{P_cF_c} \nonumber\\
    &\quad + \frac{1}{2d}\sum_{\mu,\mu'} \Bigl( F_{\mu, P_t} \otimes F_{\mu, A_1}^* \otimes F_{\mu', B_1} \otimes F_{\mu', A_2}^* \otimes \mathds{1}_{B_2F_tP_cF_c} \nonumber\\
    &\qquad\qquad + F_{\mu, P_t} \otimes \mathds{1}_{A_1} \otimes F_{\mu, B_1} \otimes F_{\mu', A_2}^* \otimes F_{\mu', B_2}^* \otimes \mathds{1}_{F_tP_cF_c}\Bigr).
\end{align}
No term in $R^{(0)}$ matches \eqref{app:eq:lcCond3}, so the flippable quantum \texttt{SWITCH} belongs to the class of bistochastic process matrices.

Finally, one finds explicit traceless contributions forbidden by ordinary process matrices. For example,
\begin{equation}
    \frac{1}{2d^2} F_{\mu, P_t} \otimes \mathds{1}_{A_1B_1A_2} \otimes F_{\mu, B_2}^* \otimes \mathds{1}_{F_tP_cF_c}
    \in \mathsf{Hrm}(\mathcal{H}_{P})\otimes\mathsf{I}_{A_1B_1A_2}\otimes\mathsf{Trl}(\mathcal{H}_{B_2})\otimes\mathsf{I}_F,
\end{equation}
as in \eqref{app:eq:diffBSP}. Therefore, the flippable quantum \texttt{SWITCH} is a valid bistochastic process matrix of type $\mathtt{BSP}_2$, but it cannot be implemented within ordinary process matrices.

% --------------------------------------------------------------------------
\subsection{Liu--Chiribella (LC) processes}
\label{app:procLC}

LC processes are higher-order operations that transform two maps into unit probability~\cite{Liu2025}. Therefore, we consider the bipartite bistochastic process-matrix type with trivial global input and output,
\begin{equation}
    \mathtt{BSP}_2^{P,F=I} 
    = \overline{(\hat{A}_1 \rightarrow \hat{B}_1) \otimes (\hat{A}_2 \rightarrow \hat{B}_2)}.
\end{equation}

\subsubsection{LC $(2,3,?)$-process}

\begin{definition}
For $n\in\{2,3\}$, the LC $(2,3,?)$ process is defined by the operator
\begin{equation}
    R_n = \sum_{i,j,k,l=0}^{n-1} p_{ijkl}\;
        |i\rangle\langle i|_{A_1}\otimes|j\rangle\langle j|_{B_1}
        \otimes|k\rangle\langle k|_{A_2}\otimes|l\rangle\langle l|_{B_2},
\end{equation}
where
\begin{equation}
    p_{ijkl} =
    \begin{cases}
      1, & \text{if } j,l\in\{0,\dots,n-1\},\ 
               i = (j+l)\!\!\!\pmod n,\ 
               k = (j-l)\!\!\!\pmod n,\\[0.2em]
      0, & \text{otherwise.}
    \end{cases}
\end{equation}
\end{definition}

In particular, there are exactly $n^2$ nonzero coefficients $p_{ijkl}$, and each contributes a rank-one projector. Hence
\begin{equation}
    \Tr[R_n] = n^2.
\end{equation}
Since $\dim(\mathcal{H}_{A_1B_1A_2B_2}) = n^4$, we have
\begin{equation}
    R_n = \lambda_{R_n}\,\mathds{1}_{A_1B_1A_2B_2} + X_{R_n},\qquad
    \lambda_{R_n} = \frac{\Tr[R_n]}{n^4} = \frac{1}{n^2},
\end{equation}
so the normalization~\eqref{app:eq:lcCond1} holds for $R_n$:
\begin{equation}
    R_n = \frac{1}{n^2}\,\mathds{1}_{A_1B_1A_2B_2} + X_{R_n}.
\end{equation}
To characterize $X_{R_n}$, we introduce traceless local diagonal operators and expand $R_n$ explicitly.

\paragraph{Case $n=2$.}

Here the nonzero coefficients correspond to the four tuples
\begin{equation}
 (i,j,k,l)\in\{(0,0,0,0),(1,1,0,0),(1,0,1,1),(0,1,1,1)\},
\end{equation}
so
\begin{equation}
\begin{aligned}
R_2
&= |0\rangle\langle 0|_{A_1}\otimes|0\rangle\langle 0|_{B_1}
   \otimes|0\rangle\langle 0|_{A_2}\otimes|0\rangle\langle 0|_{B_2}\\
&\quad + |1\rangle\langle 1|_{A_1}\otimes|1\rangle\langle 1|_{B_1}
   \otimes|0\rangle\langle 0|_{A_2}\otimes|0\rangle\langle 0|_{B_2}\\
&\quad + |1\rangle\langle 1|_{A_1}\otimes|0\rangle\langle 0|_{B_1}
   \otimes|1\rangle\langle 1|_{A_2}\otimes|1\rangle\langle 1|_{B_2}\\
&\quad + |0\rangle\langle 0|_{A_1}\otimes|1\rangle\langle 1|_{B_1}
   \otimes|1\rangle\langle 1|_{A_2}\otimes|1\rangle\langle 1|_{B_2}.
\end{aligned}
\end{equation}
Define, for each subsystem $K\in\{A_1,B_1,A_2,B_2\}$,
\begin{equation}
    X_K \coloneqq |0\rangle\langle 0|_K - \frac{1}{2}\mathds{1}_K = \frac{1}{2}\sigma^z_K ,\qquad
    Y_K \coloneqq |1\rangle\langle 1|_K - \frac{1}{2}\mathds{1}_K = -\frac{1}{2}\sigma^z_K,
\end{equation}
where $\sigma^z_K$ is the Pauli-$Z$ operator on $K$, so
\begin{equation}
|0\rangle\langle 0|_K = \frac{1}{2}\mathds{1}_K + X_K,\qquad 
|1\rangle\langle 1|_K = \frac{1}{2}\mathds{1}_K + Y_K.
\end{equation}
Writing
\begin{equation}
    R_2 = \sum_{k=1}^4 E_k\otimes F_k,
\end{equation}
with $E_k$ acting on $\mathcal{H}_{A_1}\otimes\mathcal{H}_{B_1}$ and $F_k$ on $\mathcal{H}_{A_2}\otimes\mathcal{H}_{B_2}$,
we decompose
\begin{equation}
    E_k = E_k^{(0)} + E_k^{(1)},\qquad
    F_k = F_k^{(0)} + F_k^{(1)},
\end{equation}
where the terms $E_k^{(0)} \otimes F_k^{(0)}$ collect only the components appearing in~\eqref{app:eq:lcCond3}:
\begin{equation}
\begin{aligned}
E_1^{(0)} &= \frac{1}{4}\mathds{1}_{A_1B_1} 
  - Y_{A_1}\otimes X_{B_1}, &
F_1^{(0)} &= \frac{1}{4}\mathds{1}_{A_2B_2}
  + X_{A_2}\otimes X_{B_2}, \\
E_2^{(0)} &= \frac{1}{4}\mathds{1}_{A_1B_1} 
  - Y_{A_1}\otimes Y_{B_1}, &
F_2^{(0)} &= \frac{1}{4}\mathds{1}_{A_2B_2}
  + X_{A_2}\otimes Y_{B_2}, \\
E_3^{(0)} &= \frac{1}{4}\mathds{1}_{A_1B_1} 
  + Y_{A_1}\otimes X_{B_1}, &
F_3^{(0)} &= \frac{1}{4}\mathds{1}_{A_2B_2}
  - X_{A_2}\otimes Y_{B_2}, \\
E_4^{(0)} &= \frac{1}{4}\mathds{1}_{A_1B_1} 
  + Y_{A_1}\otimes Y_{B_1}, &
F_4^{(0)} &= \frac{1}{4}\mathds{1}_{A_2B_2}
  - X_{A_2}\otimes X_{B_2}.
\end{aligned}
\end{equation}
A direct substitution then yields
\begin{equation}
    \sum_{k=1}^4 E_k^{(0)}\otimes F_k^{(0)} = \frac{1}{4}\,\mathds{1}_{A_1B_1A_2B_2},
\end{equation}
so that
\begin{equation}
R_2 = \frac{1}{4}\,\mathds{1}_{A_1B_1A_2B_2} + X_{R_2},\qquad
X_{R_2}\in \overline{\Delta}_{(\hat{A}_1 \rightarrow \hat{B}_1) \otimes (\hat{A}_2 \rightarrow \hat{B}_2)},
\end{equation}
i.e., $R_2$ is a bistochastic process matrix. Finally, one can identify explicit traceless contributions of $X_{R_2}$ that lie in~\eqref{app:eq:diffBSP}. In particular, one finds a nonzero term of the form
\begin{equation}
   \frac{1}{4}\sigma^z_{A_1}\otimes \sigma^z_{B_1}\otimes \mathds{1}_{A_2}\otimes \sigma^z_{B_2} \in \mathsf{Trl}(\mathcal{H}_{A_1})\otimes\mathsf{Trl}(\mathcal{H}_{B_1})\otimes\mathsf{I}_{A_2}\otimes\mathsf{Trl}(\mathcal{H}_{B_2}).
\end{equation}
This subspace is contained in $\Delta_{\mathtt{BSP}_2^{P,F=I} \setminus \mathtt{P}_2^{P,F=I}}$; hence, $R_2$ is a valid bistochastic process matrix, but it cannot arise as an ordinary process matrix.

\paragraph{Case $n=3$.}

The nonzero coefficients correspond to all pairs $(j,l)\in\{0,1,2\}^2$ with
\begin{equation}
i = (j+l)\!\!\!\pmod 3,\qquad k = (j-l)\!\!\!\pmod 3.
\end{equation}
Thus $R_3$ contains $9$ rank-one terms:
\begin{equation}
\begin{aligned}
R_3
&= |0\rangle\langle 0|_{A_1}\otimes|0\rangle\langle 0|_{B_1}\otimes|0\rangle\langle 0|_{A_2}\otimes|0\rangle\langle 0|_{B_2}\\
&\quad + |1\rangle\langle 1|_{A_1}\otimes|1\rangle\langle 1|_{B_1}\otimes|2\rangle\langle 2|_{A_2}\otimes|1\rangle\langle 1|_{B_2}\\
&\quad + |2\rangle\langle 2|_{A_1}\otimes|2\rangle\langle 2|_{B_1}\otimes|1\rangle\langle 1|_{A_2}\otimes|2\rangle\langle 2|_{B_2}\\
&\quad + |1\rangle\langle 1|_{A_1}\otimes|0\rangle\langle 0|_{B_1}\otimes|1\rangle\langle 1|_{A_2}\otimes|2\rangle\langle 2|_{B_2}\\
&\quad + |2\rangle\langle 2|_{A_1}\otimes|1\rangle\langle 1|_{B_1}\otimes|0\rangle\langle 0|_{A_2}\otimes|1\rangle\langle 1|_{B_2}\\
&\quad + |0\rangle\langle 0|_{A_1}\otimes|2\rangle\langle 2|_{B_1}\otimes|2\rangle\langle 2|_{A_2}\otimes|0\rangle\langle 0|_{B_2}\\
&\quad + |2\rangle\langle 2|_{A_1}\otimes|0\rangle\langle 0|_{B_1}\otimes|2\rangle\langle 2|_{A_2}\otimes|1\rangle\langle 1|_{B_2}\\
&\quad + |0\rangle\langle 0|_{A_1}\otimes|1\rangle\langle 1|_{B_1}\otimes|1\rangle\langle 1|_{A_2}\otimes|0\rangle\langle 0|_{B_2}\\
&\quad + |1\rangle\langle 1|_{A_1}\otimes|2\rangle\langle 2|_{B_1}\otimes|0\rangle\langle 0|_{A_2}\otimes|2\rangle\langle 2|_{B_2}.
\end{aligned}
\end{equation}

For each subsystem $K\in\{A_1,B_1,A_2,B_2\}$, define three traceless diagonal operators
\begin{equation}
X_K \coloneqq |0\rangle\langle 0|_K - \frac{1}{3}\mathds{1}_K,\quad
Y_K \coloneqq |1\rangle\langle 1|_K - \frac{1}{3}\mathds{1}_K,\quad
Z_K \coloneqq |2\rangle\langle 2|_K - \frac{1}{3}\mathds{1}_K,
\end{equation}
so that
\begin{equation}
|0\rangle\langle 0|_K = \frac{1}{3}\mathds{1}_K + X_K,\quad
|1\rangle\langle 1|_K = \frac{1}{3}\mathds{1}_K + Y_K,\quad
|2\rangle\langle 2|_K = \frac{1}{3}\mathds{1}_K + Z_K,\quad
X_K + Y_K + Z_K = 0.
\end{equation}

We write the decomposition
\begin{equation}
    R_3 = \sum_{k=1}^9 E_k\otimes F_k,
\end{equation}
where $E_k$ acts on $A_1B_1$ and $F_k$ acts on $A_2B_2$, and decompose each operator as
\begin{equation}
    E_k = E_k^{(0)} + E_k^{(1)},\qquad
    F_k = F_k^{(0)} + F_k^{(1)},
\end{equation}
where the terms $E_k^{(0)} \otimes F_k^{(0)}$ collect only the components appearing in~\eqref{app:eq:lcCond3}:
\begin{equation}
\begin{aligned}
E_1^{(0)} &= \frac{1}{9}\mathds{1}_{A_1B_1} + X_{A_1}\otimes X_{B_1}, &
F_1^{(0)} &= \frac{1}{9}\mathds{1}_{A_2B_2} + X_{A_2}\otimes X_{B_2},\\
E_2^{(0)} &= \frac{1}{9}\mathds{1}_{A_1B_1} + Y_{A_1}\otimes Y_{B_1}, &
F_2^{(0)} &= \frac{1}{9}\mathds{1}_{A_2B_2} + Z_{A_2}\otimes Y_{B_2},\\
E_3^{(0)} &= \frac{1}{9}\mathds{1}_{A_1B_1} + Z_{A_1}\otimes Z_{B_1}, &
F_3^{(0)} &= \frac{1}{9}\mathds{1}_{A_2B_2} + Y_{A_2}\otimes Z_{B_2},\\
E_4^{(0)} &= \frac{1}{9}\mathds{1}_{A_1B_1} + Y_{A_1}\otimes X_{B_1}, &
F_4^{(0)} &= \frac{1}{9}\mathds{1}_{A_2B_2} + Y_{A_2}\otimes Z_{B_2},\\
E_5^{(0)} &= \frac{1}{9}\mathds{1}_{A_1B_1} + Z_{A_1}\otimes Y_{B_1}, &
F_5^{(0)} &= \frac{1}{9}\mathds{1}_{A_2B_2} + X_{A_2}\otimes Y_{B_2},\\
E_6^{(0)} &= \frac{1}{9}\mathds{1}_{A_1B_1} + X_{A_1}\otimes Z_{B_1}, &
F_6^{(0)} &= \frac{1}{9}\mathds{1}_{A_2B_2} + Z_{A_2}\otimes X_{B_2},\\
E_7^{(0)} &= \frac{1}{9}\mathds{1}_{A_1B_1} + Z_{A_1}\otimes X_{B_1}, &
F_7^{(0)} &= \frac{1}{9}\mathds{1}_{A_2B_2} + Z_{A_2}\otimes Y_{B_2},\\
E_8^{(0)} &= \frac{1}{9}\mathds{1}_{A_1B_1} + X_{A_1}\otimes Y_{B_1}, &
F_8^{(0)} &= \frac{1}{9}\mathds{1}_{A_2B_2} + Y_{A_2}\otimes X_{B_2},\\
E_9^{(0)} &= \frac{1}{9}\mathds{1}_{A_1B_1} + Y_{A_1}\otimes Z_{B_1}, &
F_9^{(0)} &= \frac{1}{9}\mathds{1}_{A_2B_2} + X_{A_2}\otimes Z_{B_2}.
\end{aligned}
\end{equation}

Summing all contributions, we obtain
\begin{equation}
    \sum_{t=1}^9 E_t^{(0)}\otimes F_t^{(0)} = \frac{1}{9}\,\mathds{1}_{A_1B_1A_2B_2}.
\end{equation}

Thus
\begin{equation}
R_3 = \frac{1}{9}\mathds{1}_{A_1B_1A_2B_2} + X_{R_3},\qquad
X_{R_3}\in \overline{\Delta}_{(\hat{A}_1 \rightarrow \hat{B}_1) \otimes (\hat{A}_2 \rightarrow \hat{B}_2)},
\end{equation}
so $R_3$ is a valid deterministic bistochastic process matrix of type $\mathtt{BSP}_2^{P,F=I}$. On the other hand, one finds a nonzero term of the form
\begin{equation}
   \frac{1}{4} (X_{A_1} - Y_{A_1}) \otimes (X_{B_1} - Y_{B_1}) \otimes \mathds{1}_{A_2}\otimes (X_{B_2} - Y_{B_2}) \in \mathsf{Trl}(\mathcal{H}_{A_1})\otimes\mathsf{Trl}(\mathcal{H}_{B_1})\otimes\mathsf{I}_{A_2}\otimes\mathsf{Trl}(\mathcal{H}_{B_2}).
\end{equation}
This subspace is contained in $\Delta_{\mathtt{BSP}_2^{P,F=I} \setminus \mathtt{P}_2^{P,F=I}}$; hence, $R_3$ cannot arise as an ordinary process matrix.

\subsubsection{LC $(2,2,?)$-process}
\begin{definition}
The LC $(2,2,?)$ process is defined by the operator
\begin{equation}\label{eq:R-ABAB-short}
\begin{aligned}
R
&= (\mathds{1}-|y\rangle\langle y|)_{A_1}\otimes |x\rangle\langle x|_{B_1}\otimes |x\rangle\langle x|_{A_2}\otimes |x\rangle\langle x|_{B_2}\\
&\quad + (\mathds{1}-|y\rangle\langle y|)_{A_1}\otimes |y\rangle\langle y|_{B_1}\otimes |x\rangle\langle x|_{A_2}\otimes |y\rangle\langle y|_{B_2}\\
&\quad + |y\rangle\langle y|_{A_1}\otimes |x\rangle\langle x|_{B_1}\otimes (\mathds{1}-|x\rangle\langle x|)_{A_2}\otimes |y\rangle\langle y|_{B_2}\\
&\quad + |y\rangle\langle y|_{A_1}\otimes |y\rangle\langle y|_{B_1}\otimes (\mathds{1}-|x\rangle\langle x|)_{A_2}\otimes |x\rangle\langle x|_{B_2}\\
&\quad + |y\rangle\langle y|_{A_1}\otimes (\mathds{1}-|x\rangle\langle x|-|y\rangle\langle y|)_{B_1}
\otimes |x\rangle\langle x|_{A_2}\otimes (\mathds{1}-|x\rangle\langle x|-|y\rangle\langle y|)_{B_2},
\end{aligned}
\end{equation}
where $x,y\in\{0,\dots,d-1\}$.
\end{definition}

To express $R$ in the form~\eqref{app:eq:lcCond1}, define for each subsystem
\begin{equation}
    X_K \coloneqq |x\rangle\langle x|_K - \frac{1}{d}\mathds{1}_K,\qquad
    Y_K \coloneqq |y\rangle\langle y|_K - \frac{1}{d}\mathds{1}_K,
\end{equation}
so
\begin{equation}
|x\rangle\langle x|_K = \frac{1}{d}\mathds{1}_K + X_K,\qquad 
|y\rangle\langle y|_K = \frac{1}{d}\mathds{1}_K + Y_K,
\end{equation}
and similarly for the complements. Writing
\begin{equation}
    R = \sum_{k=1}^5 E_k\otimes F_k,
\end{equation}
with $E_k$ acting on $\mathcal{H}_{A_1} \otimes \mathcal{H}_{B_1}$ and $F_k$ on $\mathcal{H}_{A_2} \otimes \mathcal{H}_{B_2}$, we decompose
\begin{equation}
E_k = E_k^{(0)} + E_k^{(1)},\qquad F_k = F_k^{(0)} + F_k^{(1)},
\end{equation}
where the components $E_k^{(0)}\otimes F_k^{(0)}$ lie in the subspaces appearing in~\eqref{app:eq:lcCond3}:
\begin{equation}
\begin{aligned}
E_1^{(0)} &= \frac{d-1}{d^2}\mathds{1}_{A_1B_1} - Y_{A_1} \otimes X_{B_1},&
F_1^{(0)} &= \frac{1}{d^2}\mathds{1}_{A_2B_2} + X_{A_2} \otimes X_{B_2},\\
E_2^{(0)} &= \frac{d-1}{d^2}\mathds{1}_{A_1B_1} - Y_{A_1} \otimes Y_{B_1},&
F_2^{(0)} &= \frac{1}{d^2}\mathds{1}_{A_2B_2} + X_{A_2} \otimes Y_{B_2},\\
E_3^{(0)} &= \frac{1}{d^2}\mathds{1}_{A_1B_1} + Y_{A_1} \otimes X_{B_1},&
F_3^{(0)} &= \frac{d-1}{d^2}\mathds{1}_{A_2B_2} - X_{A_2} \otimes Y_{B_2},\\
E_4^{(0)} &= \frac{1}{d^2}\mathds{1}_{A_1B_1} + Y_{A_1} \otimes Y_{B_1},&
F_4^{(0)} &= \frac{d-1}{d^2}\mathds{1}_{A_2B_2} - X_{A_2} \otimes X_{B_2},\\
E_5^{(0)} &= \frac{d-2}{d^2}\mathds{1}_{A_1B_1} - Y_{A_1} \otimes X_{B_1} - Y_{A_1} \otimes Y_{B_1},&
F_5^{(0)} &= \frac{d-2}{d^2}\mathds{1}_{A_2B_2} - X_{A_2} \otimes X_{B_2} - X_{A_2} \otimes Y_{B_2}.
\end{aligned}
\end{equation}
A direct substitution yields
\begin{equation}
    \sum_{k=1}^5 E_k^{(0)}\otimes F_k^{(0)} = \frac{1}{d^2}\,\mathds{1}_{A_1B_1A_2B_2},
\end{equation}
thus establishing the required form:
\begin{equation}
R = \frac{1}{d^2}\mathds{1}_{A_1B_1A_2B_2} + X_R,\qquad X_R \in \overline{\Delta}_{(\hat{A}_1\to \hat{B}_1)\otimes(\hat{A}_2\to\hat{B}_2)}.
\end{equation}
Finally, one finds explicit traceless contributions forbidden by ordinary process matrices, e.g.
\begin{equation}
    \frac{1}{d^2}\mathds{1}_{A_1B_1A_2}\otimes(X_{B_2}+Y_{B_2}) \in \mathsf{I}_{A_1B_1A_2}\otimes\mathsf{Trl}(\mathcal{H}_{B_2}),
\end{equation}
which lies in $\Delta_{\mathtt{BSP}_2^{P,F=I} \setminus \mathtt{P}_2^{P,F=I}}$. Thus, the LC $(2,2,?)$ process is a valid bistochastic process matrix but cannot arise as an ordinary process matrix.

\end{document}